\documentclass[a4paper,UKenglish, autoref, thm-restate,11pt]{llncs}
\usepackage[a4paper,margin=2.5cm]{geometry}
\usepackage{physics}
\usepackage{amsmath}
\usepackage{tikz}
\usepackage{mathdots}
\usepackage{yhmath}
\usepackage{cancel}
\usepackage{color}
\usepackage{url}
\usepackage{hyperref}
\usepackage{cleveref}
\usepackage{siunitx}
\usepackage{array}
\usepackage{multirow}
\usepackage{amssymb}
\usepackage{gensymb}
\usepackage{tabularx}
\usepackage{extarrows}
\usepackage{booktabs}
\usetikzlibrary{fadings}
\usetikzlibrary{patterns}
\usetikzlibrary{shadows.blur}
\usetikzlibrary{shapes}

\newcommand{\N}{\mathbb{N}}

\newcommand{\equality}{\textsc{Equality}}
\newcommand{\spanningtree}{\textsc{Spanning-Tree}}
\newcommand{\sizeofG}{\textsc{Size}}
\newcommand{\stpath}{s,t-\textsc{Path}}
\newcommand{\permutation}{\textsc{Permutation}}

\newcommand{\COP}{\textsc{Corresponding Order}}

\newcommand{\hide}[1]{ }

\newcommand{\permutationgraph}{\textsc{Permutation-Recognition}}
\newcommand{\trapezoid}{\textsc{Trapezoid-Recognition}}
\newcommand{\circlegraph}{\textsc{Circle-Recognition}}
\newcommand{\polygoncircle}{k\textsc{-Polygon-Circle-Recognition}}

\newcommand{\id}{\ensuremath{\mathsf{id}}}

\newcommand{\prot}{\ensuremath{\mathsf{prot}}}

\newcommand{\dM}{\ensuremath{\mathsf{dM}}}
\newcommand{\dAM}{\ensuremath{\mathsf{dAM}}}
\newcommand{\dMA}{\ensuremath{\mathsf{dMA}}}
\newcommand{\dMAM}{\ensuremath{\mathsf{dMAM}}}

\newcommand{\dAMAM}{\ensuremath{\mathsf{dAMAM}}}
\newcommand{\dMAMA}{\ensuremath{\mathsf{dMAMA}}}
\newcommand{\cL}{\mathcal{L}}
\newcommand{\cO}{\mathcal{O}}

\usepackage{enumitem}

\newcommand{\eps}{\varepsilon}

\newcommand{\soundness}{\textbf{\textsf{Soundness. }}}
\newcommand{\completeness}{\textbf{\textsf{Completeness. }}}

\newcommand{\lipics}[1]{\textcolor{gray}{\textbf{\textsf{#1}}}}

\usepackage{graphicx}
\begin{document}
\title{Distributed Interactive Proofs for the Recognition of Some Geometric Intersection Graph Classes\thanks{This work was supported by Centro de Modelamiento Matem\'atico (CMM), ACE210010 and FB210005, BASAL funds for centers of excellence from ANID-Chile, FONDECYT 11190482 and FONDECYT 1170021.}}
%
%
\author{Benjamin Jauregui\inst{1} \and
Pedro Montealegre\inst{2} \and
Ivan Rapaport\inst{3}}
\authorrunning{F. Author et al.}

%
\institute{Departamento de Ingenier\'ia Matem\'atica, Universidad de Chile, Chile. \texttt{bjauregui@dim.uchile.cl} \and
Facultad de Ingenier\'{\i}a y Ciencias, Universidad Adolfo Iba\~nez,  Chile. \texttt{p.montealegre@uai.cl} \and
DIM-CMM (UMI 2807 CNRS), Universidad de Chile, Chile. \texttt{rapaport@dim.uchile.cl}}
\maketitle              
\begin{abstract}
A graph $G=(V,E)$ is a geometric intersection graph if every node $v \in V$ is identified with  a geometric object of some particular type, 
and two nodes are adjacent if the corresponding objects intersect. Geometric intersection graph classes have been studied from both the  theoretical and practical point of view. On the one hand, many hard problems can be efficiently solved or approximated when the input graph is restricted to a geometric
intersection class of graphs.  On the other hand, these  graphs appear naturally in many applications such as sensor networks, scheduling problems, and others. 
Recently, in the context of distributed certification and distributed interactive proofs, the recognition of graph classes has started to be intensively studied. Different results related to the recognition of trees, bipartite graphs,  bounded diameter graphs, triangle-free graphs, planar graphs,  bounded genus graphs, $H$-minor free graphs, etc., have been obtained.

The goal of the present work is to  design efficient distributed protocols for the recognition of relevant geometric intersection graph classes, namely permutation graphs, trapezoid graphs, circle graphs and polygon-circle graphs.  More precisely, for the two first classes we give proof labeling schemes recognizing them with logarithmic-sized certificates. For the other two classes, we give three-round distributed interactive protocols that use messages and certificates of size $\mathcal{O}(\log n)$. Finally, we provide  logarithmic lower-bounds on the size of the certificates on the proof labeling schemes for the recognition of any of the aforementioned  geometric intersection graph classes. 

\end{abstract}


\section{Introduction}\label{sec:intro}

This paper deals with the problem of designing compact distributed certificates and compact distributed interactive proofs
for deciding graph properties. In these protocols,
 the nodes of a connected graph $G$ have to decide, collectively, whether $G$ itself belongs to a particular graph class.
As in the centralized case, also in the distributed setting there exists a number of algorithms specially designed to decide whether 
$G$ belongs to a particular graph class. The specific goal of this work is to decide, through proof-labeling schemes
and the more general model of distributed interactive proofs, whether $G$ belongs
to relevant intersection graph classes. These classes have applications in topics like biology, ecology, computing, matrix analysis, circuit design, statistics, archaeology, etc.
For a nice survey we refer to~\cite{McKee1999}.

 \subsection{Proof-Labeling Schemes and Distributed Interactive Proofs}
 
In locally decidable algorithms every node is just allowed  to send messages to its neighbors, in one round
 (a less restrictive, but similar scenario, is where the number of rounds is constant, independent of the size of $G$, see \cite{naor1995can}).
Some very basic properties can be decided locally (with a local algorithm).  For instance,  deciding whether the graph $G$ has bounded degree.
More generally, if we do not impose bandwidth restrictions, then detecting the existence of any local structure (such as a triangle) can be solved through 
local algorithms.

In the aforementioned examples, acceptance and rejection  are  (implicitly) defined as follows. 
If $G$ satisfies the property, then all nodes must accept; otherwise, at least one node must reject.
These very fast local algorithms could be used in distributed fault-tolerant computing, where the nodes, with some regularity,
must check whether the current network configuration is in a legal state~\cite{korman2010proof}. Then, if the configuration becomes
at some point illegal, the rejecting node(s) raise the alarm or launch a recovery procedure. When there are distributed algorithms designed for particular graph classes, the use of an initial  recognizing protocol  could avoid the risk
of running a distributed protocol for graphs that do not belong to the class for which the protocol was designed.

Of course, many simple properties cannot be decided with one round (or with a constant number of rounds) through local algorithms. In order to overcome this issue, the notion of proof-labeling scheme (PLS) was introduced~\cite{korman2010proof}. PLSs can be seen as a distributed counterpart of the nondeterministic class  NP. In fact, in a PLS, a powerful prover gives to every node $v$ a certificate $c(v)$. This provides $G$ with a global distributed proof. Then, every node $v$   performs a local verification using its local information  together with $c(v)$. 


Incorporating a powerful prover to the model is not just motivated by a purely theoretical interest.
 In fact,  with the rise of the Internet, prover-assisted computing models are ubiquitous. 
Asymmetric applications --social networks, cloud computing, etc.-- where a very powerful central entity stores and process 
large amounts of data, are already part of our everyday lives.
A key issue, which is  a central part of the PLS  model, is that the devices of the network cannot trust the central entity
and are forced to verify the correctness of the distributed proof.

The generalization of the class NP to interactive proof systems, a model where the prover and the verifier are allowed to interact
was a breakthrough in computational complexity~\cite{babai1988arthur,goldreich1991proofs,goldwasser1989knowledge,lund1992algebraic,shamir1992ip}.  In the distributed framework, the notion of distributed interactive protocols was introduced  in~\cite{kol2018interactive} and further studied in~\cite{crescenzi2019trade,fraigniaud2019distributed,montealegre2020shared,naor2020power}. In such protocols, a centralized, untrustable prover with unlimited computation power, named Merlin, exchanges messages with a randomized distributed algorithm, named Arthur. 

Let us illustrate the general idea of this model, and also some notation. If we consider, for instance, four interactions, then there are two
possible protocols:  a \dAMAM\  protocol and  \dMAMA\ protocol.  In a \dAMAM\ protocol, also denoted \dAM[4],  the last interaction is performed by Merlin.
In a \dMAMA\ protocol, also denoted \dMA[4], the last interaction is performed by Arthur. Note that \dAM[1]=\dM, and we recover exactly the PLS model.

Now, we are going to explain with more detail what happens in a particular case. Let us consider
a three interaction, \dMAM=\dAM[3] protocol. In this case Merlin starts, and he  provides a certificate to  Arthur (that is, certificates $c(v)$ for every node $v \in V$). Then, 
a random string (fresh randomness) is generated and made public to both Arthur (the nodes) and Merlin.
This is the interaction performed by Arthur, and it should be interpreted as if Arthur was challenging Merlin.
Note also that we are considering here the shared randomness setting.

Finally, Merlin replies to the query by sending another distributed certificate. After all these interactions,
comes the  deterministic distributed verification phase,  performed between every node and its neighbors,  
after which every node decides whether to accept or reject.

We say that an algorithm uses $\cO(f(n))$ bits if the messages exchanged between the nodes (in the verification round), and also the certificates 
sent by the prover Merlin to the nodes, are upper bounded by $\cO(f(n))$. We include this bandwidth bound  in the notation, which becomes $\dMA[k,f(n)] $ and $\dAM[k,f(n)]$
for the corresponding protocols.

Interaction may decrease drastically the size of the messages needed to solve some problems. 
Consider, for instance, the problem symmetry,  where the nodes are asked  to decide whether the graph $G$ has a non-trivial 
automorphism (i.e., a non-trivial one-to-one mapping from the set of nodes to itself preserving edges). 
Any PLS solving the  symmetry problem requires certificates of size $\Omega(n^2)$~\cite{goos2016locally}. Nevertheless,
this problem admits distributed interactive protocols with small certificates, and very few interactions. In fact, 
it can be solved with  both a $\dMAM[\log n]$ protocol and a $\dAM[n \log n]$ protocol~\cite{kol2018interactive}.

\subsection{Geometric Intersection Graph Classes}

A graph $G=(V,E)$ is a geometric intersection graph if every node $v \in V$ is identified with  a geometric object of some particular type, 
and two vertices are adjacent if the corresponding objects intersect. The two simplest non-trivial, and arguably two of the most studied
geometric intersection graphs are  {\bf{interval graphs}} and  {\bf{permutation graphs}}. In fact, most of the best-known geometric intersection graph classes are either generalizations of interval graphs or generalizations of permutation graphs.
It comes  as no surprise that many papers address different algorithmic and structural  aspects, simultaneously, in both interval and permutation graph
\cite{asdre2007harmonious,foucaud2017identification1,foucaud2017identification,kante2013enumeration,madanlal1996tree,kratsch2006certifying,yamazaki2020enumeration}.

In both interval and permutation graphs, the intersecting objects  are (line) segments, with different restrictions imposed on their positions.
In interval graphs, the segments must all lie on the real line. In permutation graphs, the endpoints of the segments must  lie on two separate, parallel real lines. In Figure \ref{fig:Expermutation} we show an example of a permutation graph.
	
\begin{figure}[h!]
	\centering 
	\tikzset{every picture/.style={line width=0.75pt}} 

\begin{tikzpicture}[x=0.75pt,y=0.75pt,yscale=-1,xscale=1]

\draw [color={rgb, 255:red, 208; green, 2; blue, 27 }  ,draw opacity=1 ][line width=1.5]    (205,135) -- (245,60) ;
\draw [color={rgb, 255:red, 245; green, 166; blue, 35 }  ,draw opacity=1 ][line width=1.5]    (305,135) -- (225,60) ;
\draw [color={rgb, 255:red, 248; green, 231; blue, 28 }  ,draw opacity=1 ][line width=1.5]    (245,135) -- (205,60) ;
\draw [color={rgb, 255:red, 189; green, 16; blue, 224 }  ,draw opacity=1 ][line width=1.5]    (225,135) -- (305,60) ;
\draw [color={rgb, 255:red, 74; green, 144; blue, 226 }  ,draw opacity=1 ][line width=1.5]    (265,135) -- (325,60) ;
\draw [color={rgb, 255:red, 155; green, 155; blue, 155 }  ,draw opacity=1 ][line width=1.5]    (325,135) -- (285,60) ;
\draw [color={rgb, 255:red, 126; green, 211; blue, 33 }  ,draw opacity=1 ][line width=1.5]    (285,135) -- (265,60) ;
\draw    (380,75) -- (385,120) -- (410,110) -- (415,70) -- (475,65) -- (455,90) ;
\draw    (475,65) -- (475,120) ;
\draw    (415,70) -- (380,75) ;
\draw    (475,120) -- (410,110) ;
\draw    (455,90) -- (410,110) ;
\draw    (415,70) -- (455,90) ;
\draw  [fill={rgb, 255:red, 126; green, 211; blue, 33 }  ,fill opacity=1 ] (450,90) .. controls (450,87.24) and (452.24,85) .. (455,85) .. controls (457.76,85) and (460,87.24) .. (460,90) .. controls (460,92.76) and (457.76,95) .. (455,95) .. controls (452.24,95) and (450,92.76) .. (450,90) -- cycle ;
\draw  [fill={rgb, 255:red, 248; green, 231; blue, 28 }  ,fill opacity=1 ] (380,120) .. controls (380,117.24) and (382.24,115) .. (385,115) .. controls (387.76,115) and (390,117.24) .. (390,120) .. controls (390,122.76) and (387.76,125) .. (385,125) .. controls (382.24,125) and (380,122.76) .. (380,120) -- cycle ;
\draw  [fill={rgb, 255:red, 208; green, 2; blue, 27 }  ,fill opacity=1 ] (375,75) .. controls (375,72.24) and (377.24,70) .. (380,70) .. controls (382.76,70) and (385,72.24) .. (385,75) .. controls (385,77.76) and (382.76,80) .. (380,80) .. controls (377.24,80) and (375,77.76) .. (375,75) -- cycle ;
\draw  [fill={rgb, 255:red, 189; green, 16; blue, 224 }  ,fill opacity=1 ] (405,110) .. controls (405,107.24) and (407.24,105) .. (410,105) .. controls (412.76,105) and (415,107.24) .. (415,110) .. controls (415,112.76) and (412.76,115) .. (410,115) .. controls (407.24,115) and (405,112.76) .. (405,110) -- cycle ;
\draw  [fill={rgb, 255:red, 245; green, 166; blue, 35 }  ,fill opacity=1 ] (410,70) .. controls (410,67.24) and (412.24,65) .. (415,65) .. controls (417.76,65) and (420,67.24) .. (420,70) .. controls (420,72.76) and (417.76,75) .. (415,75) .. controls (412.24,75) and (410,72.76) .. (410,70) -- cycle ;
\draw  [fill={rgb, 255:red, 155; green, 155; blue, 155 }  ,fill opacity=1 ] (470,120) .. controls (470,117.24) and (472.24,115) .. (475,115) .. controls (477.76,115) and (480,117.24) .. (480,120) .. controls (480,122.76) and (477.76,125) .. (475,125) .. controls (472.24,125) and (470,122.76) .. (470,120) -- cycle ;
\draw  [fill={rgb, 255:red, 74; green, 144; blue, 226 }  ,fill opacity=1 ] (470,65) .. controls (470,62.24) and (472.24,60) .. (475,60) .. controls (477.76,60) and (480,62.24) .. (480,65) .. controls (480,67.76) and (477.76,70) .. (475,70) .. controls (472.24,70) and (470,67.76) .. (470,65) -- cycle ;
\draw [line width=0.75]    (205,135) -- (325,135) ;
\draw [line width=0.75]    (205,60) -- (325,60) ;
\draw [line width=0.75]    (205,55) -- (205,65) ;
\draw [line width=0.75]    (225,55) -- (225,65) ;
\draw [line width=0.75]    (245,55) -- (245,65) ;
\draw [line width=0.75]    (265,55) -- (265,65) ;
\draw [line width=0.75]    (285,55) -- (285,65) ;
\draw [line width=0.75]    (305,55) -- (305,65) ;
\draw [line width=0.75]    (325,55) -- (325,65) ;
\draw [line width=0.75]    (205,130) -- (205,140) ;
\draw [line width=0.75]    (225,130) -- (225,140) ;
\draw [line width=0.75]    (245,130) -- (245,140) ;
\draw [line width=0.75]    (265,130) -- (265,140) ;
\draw [line width=0.75]    (285,130) -- (285,140) ;
\draw [line width=0.75]    (305,130) -- (305,140) ;
\draw [line width=0.75]    (325,130) -- (325,140) ;

\end{tikzpicture}
	\caption{An example of a permutation graph with its corresponding intersection model.}
	\label{fig:Expermutation}
\end{figure}
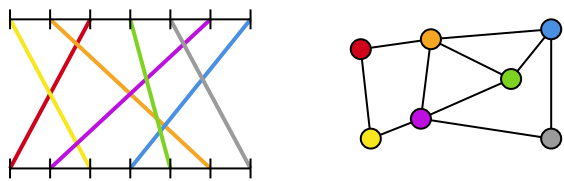

Although the class of interval graphs is quite restrictive, there are a number of practical applications and specialized algorithms for interval graphs \cite{golumbic2004algorithmic,halldorsson2020improved,konrad2019distributed}. 
Moreover, for several applications, the subclass of unit interval graphs (the situation where all the intervals have the same length)
turns out to be extremely useful as well~\cite{beeri1983desirability,kaplan1996pathwidth}.

A  natural generalization of interval graphs are {\bf{circular arc graphs}}, where the segments, instead of lying  
on a line, lie on a circle. More precisely, a circular arc graph is the intersection graph of arcs of a circle.
Although circular arc graphs look similar to interval graphs, several combinatorial problems behave very differently on these two classes of graphs.  
For example, the coloring problem is NP-complete for circular arc graphs while it can be solved in linear time on interval graphs \cite{garey1980complexity}.
Recognizing circular-arc graphs can be done in linear-time \cite{kaplan2006simpler,mcconnell2003linear}.

The class of  permutation graphs behaves  as the class of interval graphs in the sense that,
 on one hand,  permutation graphs can be recognized in linear time \cite{kratsch2006certifying} and,  on the other hand,
 many NP-complete problems can be solved efficiently when the input is restricted to permutation graphs \cite{chao2000optimal,lappas2013n}.
A graph $G$ is a {\bf{circle graph}} if $G$ is the intersection model of a collection of chords in a circle (see Figure \ref{fig:Excircle}).
Clearly, circle graphs are a generalization of permutation graphs. In fact, permutation graphs can be characterized as  circle graphs that admit an \emph{equator}, i.e., an additional chord that intersects every other chord. Circle graphs can be recognized 
in time $\cO(n^2)$ \cite{spinrad1994recognition}. Many NP-complete problems can be solve in polynomial
time when restricted to circle graphs \cite{kloks1993treewidth,tiskin2015fast}.

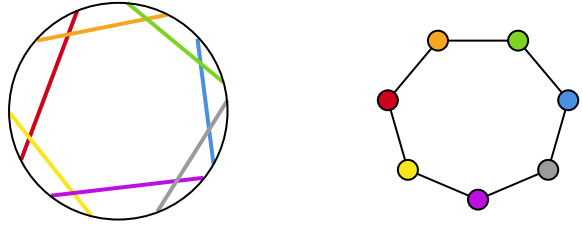
\begin{figure}[h!]
		\centering 
		\tikzset{every picture/.style={line width=0.75pt}} 

\begin{tikzpicture}[x=0.75pt,y=0.75pt,yscale=-1,xscale=1]

\draw [color={rgb, 255:red, 208; green, 2; blue, 27 }  ,draw opacity=1 ][line width=1.5]    (207,110) -- (235,35) ;
\draw [color={rgb, 255:red, 245; green, 166; blue, 35 }  ,draw opacity=1 ][line width=1.5]    (280,37) -- (215,50) ;
\draw [color={rgb, 255:red, 248; green, 231; blue, 28 }  ,draw opacity=1 ][line width=1.5]    (242,138) -- (201,86) ;
\draw [color={rgb, 255:red, 189; green, 16; blue, 224 }  ,draw opacity=1 ][line width=1.5]    (222,128) -- (298,119) ;
\draw [color={rgb, 255:red, 74; green, 144; blue, 226 }  ,draw opacity=1 ][line width=1.5]    (303,112) -- (295,49) ;
\draw [color={rgb, 255:red, 155; green, 155; blue, 155 }  ,draw opacity=1 ][line width=1.5]    (275,136) -- (310,80) ;
\draw [color={rgb, 255:red, 126; green, 211; blue, 33 }  ,draw opacity=1 ][line width=1.5]    (308,71) -- (260,31) ;
\draw    (400,115) -- (390,80) -- (415,50) -- (455,50) -- (480,80) -- (470,115) -- (435,130) -- cycle ;
\draw  [fill={rgb, 255:red, 126; green, 211; blue, 33 }  ,fill opacity=1 ] (450,50) .. controls (450,47.24) and (452.24,45) .. (455,45) .. controls (457.76,45) and (460,47.24) .. (460,50) .. controls (460,52.76) and (457.76,55) .. (455,55) .. controls (452.24,55) and (450,52.76) .. (450,50) -- cycle ;
\draw  [fill={rgb, 255:red, 248; green, 231; blue, 28 }  ,fill opacity=1 ] (395,115) .. controls (395,112.24) and (397.24,110) .. (400,110) .. controls (402.76,110) and (405,112.24) .. (405,115) .. controls (405,117.76) and (402.76,120) .. (400,120) .. controls (397.24,120) and (395,117.76) .. (395,115) -- cycle ;
\draw  [fill={rgb, 255:red, 208; green, 2; blue, 27 }  ,fill opacity=1 ] (385,80) .. controls (385,77.24) and (387.24,75) .. (390,75) .. controls (392.76,75) and (395,77.24) .. (395,80) .. controls (395,82.76) and (392.76,85) .. (390,85) .. controls (387.24,85) and (385,82.76) .. (385,80) -- cycle ;
\draw  [fill={rgb, 255:red, 189; green, 16; blue, 224 }  ,fill opacity=1 ] (430,130) .. controls (430,127.24) and (432.24,125) .. (435,125) .. controls (437.76,125) and (440,127.24) .. (440,130) .. controls (440,132.76) and (437.76,135) .. (435,135) .. controls (432.24,135) and (430,132.76) .. (430,130) -- cycle ;
\draw  [fill={rgb, 255:red, 245; green, 166; blue, 35 }  ,fill opacity=1 ] (410,50) .. controls (410,47.24) and (412.24,45) .. (415,45) .. controls (417.76,45) and (420,47.24) .. (420,50) .. controls (420,52.76) and (417.76,55) .. (415,55) .. controls (412.24,55) and (410,52.76) .. (410,50) -- cycle ;
\draw  [fill={rgb, 255:red, 155; green, 155; blue, 155 }  ,fill opacity=1 ] (465,115) .. controls (465,112.24) and (467.24,110) .. (470,110) .. controls (472.76,110) and (475,112.24) .. (475,115) .. controls (475,117.76) and (472.76,120) .. (470,120) .. controls (467.24,120) and (465,117.76) .. (465,115) -- cycle ;
\draw  [fill={rgb, 255:red, 74; green, 144; blue, 226 }  ,fill opacity=1 ] (475,80) .. controls (475,77.24) and (477.24,75) .. (480,75) .. controls (482.76,75) and (485,77.24) .. (485,80) .. controls (485,82.76) and (482.76,85) .. (480,85) .. controls (477.24,85) and (475,82.76) .. (475,80) -- cycle ;
\draw  [line width=0.75]  (201,85.5) .. controls (201,55.4) and (225.4,31) .. (255.5,31) .. controls (285.6,31) and (310,55.4) .. (310,85.5) .. controls (310,115.6) and (285.6,140) .. (255.5,140) .. controls (225.4,140) and (201,115.6) .. (201,85.5) -- cycle ;

\end{tikzpicture}
		\caption{An example of a circle graph with its corresponding intersection model.}
		\label{fig:Excircle}
	\end{figure}

The well-known class of {\bf{trapezoid graphs}}  is a  generalization of both interval graphs and permutation graphs.
A trapezoid graph  is defined as the intersection graph of trapezoids between two horizontal lines (see Figure \ref{fig:Extrapezoid}). 
Ma and Spinrad \cite{ma19942} showed that trapezoid graphs can be recognized in $\cO(n^2)$ time.
Trapezoid graphs were applied in various contexts such as  VLSI design \cite{dagan1988trapezoid} and  
bioinformatics \cite{abouelhoda2005chaining}. Note that trapezoid and circle graphs are incomparable: the  
trapezoid graph of Figure~\ref{fig:Extrapezoid}  is not a circle graph, while the circle graph of Figure~\ref{fig:Excircle} is not a trapezoid graph.
 
\begin{figure}[h!]
\centering 
\tikzset{every picture/.style={line width=0.75pt}} 

\begin{tikzpicture}[x=0.75pt,y=0.75pt,yscale=-1,xscale=1]

\draw  [color={rgb, 255:red, 74; green, 144; blue, 226 }  ,draw opacity=1 ][fill={rgb, 255:red, 74; green, 144; blue, 226 }  ,fill opacity=0.2 ] (350,60) -- (390,165) -- (310,165) -- (230,60) -- cycle ;
\draw  [color={rgb, 255:red, 189; green, 16; blue, 224 }  ,draw opacity=1 ][fill={rgb, 255:red, 189; green, 16; blue, 224 }  ,fill opacity=0.2 ] (370,60) -- (370,165) -- (270,165) -- (270,60) -- cycle ;
\draw  [color={rgb, 255:red, 248; green, 231; blue, 28 }  ,draw opacity=1 ][fill={rgb, 255:red, 248; green, 231; blue, 28 }  ,fill opacity=0.2 ] (250,60) -- (250,165) -- (170,165) -- (150,60) -- cycle ;
\draw  [color={rgb, 255:red, 245; green, 166; blue, 35 }  ,draw opacity=1 ][fill={rgb, 255:red, 245; green, 166; blue, 35 }  ,fill opacity=0.2 ] (410,60) -- (410,165) -- (330,165) -- (330,60) -- cycle ;
\draw  [color={rgb, 255:red, 128; green, 128; blue, 128 }  ,draw opacity=1 ][fill={rgb, 255:red, 128; green, 128; blue, 128 }  ,fill opacity=0.2 ] (390,60) -- (350,165) -- (230,165) -- (310,60) -- cycle ;
\draw    (480,57) -- (620,57) -- (550,167) -- cycle ;
\draw    (535,119) -- (550,167) -- (565,119) ;
\draw    (550,84) -- (480,57) -- (535,119) -- (550,107) ;
\draw    (550,84) -- (620,57) -- (565,119) -- (550,107) -- cycle ;
\draw    (535,119) -- (565,119) -- (550,84) -- cycle ;
\draw  [fill={rgb, 255:red, 126; green, 211; blue, 33 }  ,fill opacity=1 ] (475,57) .. controls (475,54.24) and (477.24,52) .. (480,52) .. controls (482.76,52) and (485,54.24) .. (485,57) .. controls (485,59.76) and (482.76,62) .. (480,62) .. controls (477.24,62) and (475,59.76) .. (475,57) -- cycle ;
\draw  [fill={rgb, 255:red, 248; green, 231; blue, 28 }  ,fill opacity=1 ] (615,57) .. controls (615,54.24) and (617.24,52) .. (620,52) .. controls (622.76,52) and (625,54.24) .. (625,57) .. controls (625,59.76) and (622.76,62) .. (620,62) .. controls (617.24,62) and (615,59.76) .. (615,57) -- cycle ;
\draw  [fill={rgb, 255:red, 208; green, 2; blue, 27 }  ,fill opacity=1 ] (545,167) .. controls (545,164.24) and (547.24,162) .. (550,162) .. controls (552.76,162) and (555,164.24) .. (555,167) .. controls (555,169.76) and (552.76,172) .. (550,172) .. controls (547.24,172) and (545,169.76) .. (545,167) -- cycle ;
\draw  [fill={rgb, 255:red, 189; green, 16; blue, 224 }  ,fill opacity=1 ] (530,119) .. controls (530,116.24) and (532.24,114) .. (535,114) .. controls (537.76,114) and (540,116.24) .. (540,119) .. controls (540,121.76) and (537.76,124) .. (535,124) .. controls (532.24,124) and (530,121.76) .. (530,119) -- cycle ;
\draw  [fill={rgb, 255:red, 245; green, 166; blue, 35 }  ,fill opacity=1 ] (545,107) .. controls (545,104.24) and (547.24,102) .. (550,102) .. controls (552.76,102) and (555,104.24) .. (555,107) .. controls (555,109.76) and (552.76,112) .. (550,112) .. controls (547.24,112) and (545,109.76) .. (545,107) -- cycle ;
\draw  [fill={rgb, 255:red, 155; green, 155; blue, 155 }  ,fill opacity=1 ] (560,119) .. controls (560,116.24) and (562.24,114) .. (565,114) .. controls (567.76,114) and (570,116.24) .. (570,119) .. controls (570,121.76) and (567.76,124) .. (565,124) .. controls (562.24,124) and (560,121.76) .. (560,119) -- cycle ;
\draw  [fill={rgb, 255:red, 74; green, 144; blue, 226 }  ,fill opacity=1 ] (545,84) .. controls (545,81.24) and (547.24,79) .. (550,79) .. controls (552.76,79) and (555,81.24) .. (555,84) .. controls (555,86.76) and (552.76,89) .. (550,89) .. controls (547.24,89) and (545,86.76) .. (545,84) -- cycle ;
\draw  [color={rgb, 255:red, 208; green, 2; blue, 27 }  ,draw opacity=1 ][fill={rgb, 255:red, 208; green, 2; blue, 27 }  ,fill opacity=0.2 ] (210,60) -- (290,165) -- (190,165) -- (170,60) -- cycle ;
\draw  [color={rgb, 255:red, 65; green, 117; blue, 5 }  ,draw opacity=1 ][fill={rgb, 255:red, 126; green, 211; blue, 33 }  ,fill opacity=0.2 ] (290,60) -- (210,165) -- (150,165) -- (190,60) -- cycle ;
\draw [line width=0.75]    (150,60) -- (401.83,60) -- (410,60) ;
\draw [line width=0.75]    (150,165) -- (410,165) ;
\draw [line width=0.75]    (290,160) -- (290,170) ;
\draw [line width=0.75]    (310,160) -- (310,170) ;
\draw [line width=0.75]    (330,160) -- (330,170) ;
\draw [line width=0.75]    (350,160) -- (350,170) ;
\draw [line width=0.75]    (370,160) -- (370,170) ;
\draw [line width=0.75]    (390,160) -- (390,170) ;
\draw [line width=0.75]    (410,160) -- (410,170) ;
\draw [line width=0.75]    (150,160) -- (150,170) ;
\draw [line width=0.75]    (170,160) -- (170,170) ;
\draw [line width=0.75]    (190,160) -- (190,170) ;
\draw [line width=0.75]    (210,160) -- (210,170) ;
\draw [line width=0.75]    (230,160) -- (230,170) ;
\draw [line width=0.75]    (250,160) -- (250,170) ;
\draw [line width=0.75]    (270,160) -- (270,170) ;
\draw [line width=0.75]    (290,55) -- (290,65) ;
\draw [line width=0.75]    (310,55) -- (310,65) ;
\draw [line width=0.75]    (330,55) -- (330,65) ;
\draw [line width=0.75]    (350,55) -- (350,65) ;
\draw [line width=0.75]    (370,55) -- (370,65) ;
\draw [line width=0.75]    (390,55) -- (390,65) ;
\draw [line width=0.75]    (410,55) -- (410,65) ;
\draw [line width=0.75]    (150,55) -- (150,65) ;
\draw [line width=0.75]    (170,55) -- (170,65) ;
\draw [line width=0.75]    (190,55) -- (190,65) ;
\draw [line width=0.75]    (210,55) -- (210,65) ;
\draw [line width=0.75]    (230,55) -- (230,65) ;
\draw [line width=0.75]    (250,55) -- (250,65) ;
\draw [line width=0.75]    (270,55) -- (270,65) ;

\end{tikzpicture}
\caption{An example of a permutation graph with its corresponding intersection model.}
\label{fig:Extrapezoid}
\end{figure}
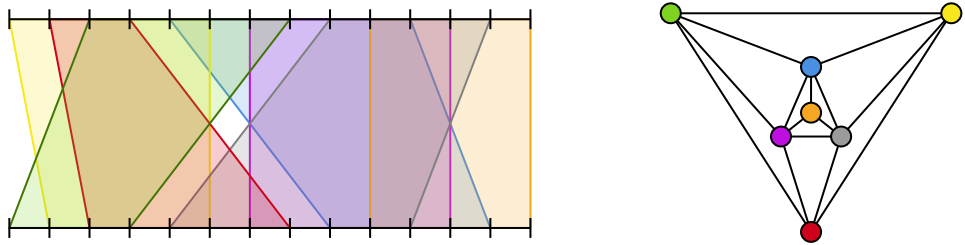

Recall that the way permutation graphs were generalized to circular graphs
is by placing the ends of the segments in a circle (as chords) instead of
placing the ends of the segments in two parallel lines.
The same approach is used to 
generalize trapezoidal graphs and thus introducing polygon circle graphs.

More precisely, a {\bf{polygon circle graph}}  is the  intersection graph of convex polygons of $k$ sides, all of whose vertices lie on  a circle.
In this  case we refer to a  $k$-polygon circle graphs. 
In Figure \ref{fig:Expolygoncircle} we show an example of a $3$-polygon circle graph.  
Both  trapezoid graphs and circle graphs are proper subclasses of polygon  circle graphs.
Note that the polygon circle graph of Figure \ref{fig:Expolygoncircle} is neither a trapezoid graph nor a circle graph. 

\begin{figure}[h!]
		\centering 
		\tikzset{every picture/.style={line width=0.75pt}} 

\begin{tikzpicture}[x=0.75pt,y=0.75pt,yscale=-1,xscale=1]

\draw    (455,55) -- (430,85) ;
\draw    (455,115) -- (430,85) ;
\draw    (430,85) -- (375,85) ;
\draw    (400,115) -- (375,85) -- (400,55) -- (455,55) -- (480,85) -- (455,115) -- cycle ;
\draw  [fill={rgb, 255:red, 126; green, 211; blue, 33 }  ,fill opacity=1 ] (450,55) .. controls (450,52.24) and (452.24,50) .. (455,50) .. controls (457.76,50) and (460,52.24) .. (460,55) .. controls (460,57.76) and (457.76,60) .. (455,60) .. controls (452.24,60) and (450,57.76) .. (450,55) -- cycle ;
\draw  [fill={rgb, 255:red, 248; green, 231; blue, 28 }  ,fill opacity=1 ] (395,115) .. controls (395,112.24) and (397.24,110) .. (400,110) .. controls (402.76,110) and (405,112.24) .. (405,115) .. controls (405,117.76) and (402.76,120) .. (400,120) .. controls (397.24,120) and (395,117.76) .. (395,115) -- cycle ;
\draw  [fill={rgb, 255:red, 208; green, 2; blue, 27 }  ,fill opacity=1 ] (370,85) .. controls (370,82.24) and (372.24,80) .. (375,80) .. controls (377.76,80) and (380,82.24) .. (380,85) .. controls (380,87.76) and (377.76,90) .. (375,90) .. controls (372.24,90) and (370,87.76) .. (370,85) -- cycle ;
\draw  [fill={rgb, 255:red, 74; green, 144; blue, 226 }  ,fill opacity=1 ] (425,85) .. controls (425,82.24) and (427.24,80) .. (430,80) .. controls (432.76,80) and (435,82.24) .. (435,85) .. controls (435,87.76) and (432.76,90) .. (430,90) .. controls (427.24,90) and (425,87.76) .. (425,85) -- cycle ;
\draw  [fill={rgb, 255:red, 245; green, 166; blue, 35 }  ,fill opacity=1 ] (395,55) .. controls (395,52.24) and (397.24,50) .. (400,50) .. controls (402.76,50) and (405,52.24) .. (405,55) .. controls (405,57.76) and (402.76,60) .. (400,60) .. controls (397.24,60) and (395,57.76) .. (395,55) -- cycle ;
\draw  [fill={rgb, 255:red, 155; green, 155; blue, 155 }  ,fill opacity=1 ] (450,115) .. controls (450,112.24) and (452.24,110) .. (455,110) .. controls (457.76,110) and (460,112.24) .. (460,115) .. controls (460,117.76) and (457.76,120) .. (455,120) .. controls (452.24,120) and (450,117.76) .. (450,115) -- cycle ;
\draw  [fill={rgb, 255:red, 189; green, 16; blue, 224 }  ,fill opacity=1 ] (475,85) .. controls (475,82.24) and (477.24,80) .. (480,80) .. controls (482.76,80) and (485,82.24) .. (485,85) .. controls (485,87.76) and (482.76,90) .. (480,90) .. controls (477.24,90) and (475,87.76) .. (475,85) -- cycle ;
\draw  [color={rgb, 255:red, 208; green, 2; blue, 27 }  ,draw opacity=1 ][fill={rgb, 255:red, 208; green, 2; blue, 27 }  ,fill opacity=0.2 ] (220,45) -- (225,130) -- (201,85.5) -- cycle ;
\draw  [color={rgb, 255:red, 245; green, 166; blue, 35 }  ,draw opacity=1 ][fill={rgb, 255:red, 245; green, 166; blue, 35 }  ,fill opacity=0.2 ] (275,35) -- (204,70) -- (213,52) -- cycle ;
\draw  [color={rgb, 255:red, 248; green, 231; blue, 28 }  ,draw opacity=1 ][fill={rgb, 255:red, 248; green, 231; blue, 28 }  ,fill opacity=0.2 ] (284,132) -- (205,105) -- (215,121) -- cycle ;
\draw  [color={rgb, 255:red, 128; green, 128; blue, 128 }  ,draw opacity=1 ][fill={rgb, 255:red, 128; green, 128; blue, 128 }  ,fill opacity=0.2 ] (310,86.5) -- (255.5,140) -- (309,95) -- cycle ;
\draw  [color={rgb, 255:red, 126; green, 211; blue, 33 }  ,draw opacity=1 ][fill={rgb, 255:red, 126; green, 211; blue, 33 }  ,fill opacity=0.2 ] (303,60) -- (242,33) -- (307,68) -- cycle ;
\draw  [color={rgb, 255:red, 189; green, 16; blue, 224 }  ,draw opacity=1 ][fill={rgb, 255:red, 189; green, 16; blue, 224 }  ,fill opacity=0.2 ] (293,46.5) -- (297,120) -- (309,79) -- cycle ;
\draw  [color={rgb, 255:red, 74; green, 144; blue, 226 }  ,draw opacity=1 ][fill={rgb, 255:red, 74; green, 144; blue, 226 }  ,fill opacity=0.2 ] (285,40) -- (290,127) -- (201,85.5) -- cycle ;
\draw  [line width=0.75]  (201,85.5) .. controls (201,55.4) and (225.4,31) .. (255.5,31) .. controls (285.6,31) and (310,55.4) .. (310,85.5) .. controls (310,115.6) and (285.6,140) .. (255.5,140) .. controls (225.4,140) and (201,115.6) .. (201,85.5) -- cycle ;

\end{tikzpicture}
		\caption{An example of a $3$-polygon circle graph with its corresponding intersection model.}
		\label{fig:Expolygoncircle}
	\end{figure}
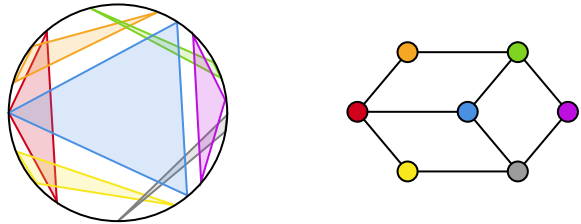	

The problem of recognizing whether a graph is a $k$-polygon circle graphs, for any $k \geq 3$, is NP-complete \cite{pergel2007recognition}.
Nevertheless, many NP-complete problems have polynomial time algorithms when restricted to polygon circle graphs \cite{garey1980complexity,gavril2000maximum}.
 In an unpublished result, M. Fellows proved that the class of polygon circle graphs is closed under taking induced minors.

\subsection{Our Results}

%
 In Section  \ref{sec:tool} we recall, explain and develop some tools, building blocks for the rest of the paper.
In  Sections \ref{sec:permutation}  and  \ref{sec:trapezoid} we prove that permutation and trapezoid graphs can be recognized by  PLSs
with certificates of size $\cO(\log n)$, and then we obtain the result for permutation graphs as a corollary. 
Then, in Sections \ref{sec:circle} and \ref{sec:polyg} we prove that circle and $k$-polygon circle graphs can  be
recognized with a three round, \dMAM\ protocol, and we obtain the result for circle graphs as a particular case when $k=2$.
Finally, in Section  \ref{sec:lower}, we prove that any PLS for recognizing permutation graphs, trapezoid graphs, circle graphs or
polygon circle graphs requires certificates of size  $\Omega(\log n)$.

\subsection{Related Work}

We already know the existence of PLSs  (with logarithmic size certificates) for the recognition of many graph classes such as 
acyclic graphs \cite{korman2010proof}, planar graphs~\cite{feuilloley2020compact}, graphs with bounded genus~\cite{feuilloley2020local},
graph classes defined by a finite set of forbidden minors~\cite{bousquet2021local}, etc.

The distributed interactive proof  model is clearly more powerful than the PLS model.
Some problems,  for which any PLS  requires huge certificates, can  be solved with a distributed interactive protocol
with small certificates and, in fact, with very few interactions. This is the case  of problem symmetry, where the system 
must decide  whether a graphs has a non-trivial automorphism. Problem symmetry is  in $\dMAM[\log n]$ and  also in $\dAM[n \log n]$,
while the size of any certificate in a PLS  must be of size at least $\Omega(n^2)$~\cite{kol2018interactive}.

It is always important to keep in mind
 the compiler defined in~\cite{naor2020power} which turns, automatically, any problem solved in NP in time $\tau(n)$ into a $\dMAM$ protocol
 that uses bandwidth $\tau(n) \log n/n$. Therefore, any class of sparse 
 graphs that can be recognized in linear time, can also be recognized by a $\dMAM$ protocol with logarithmic-sized certificates.

 Any geometric intersection graph class is hereditary
(a graph class is hereditary if the class is closed under taking induced subgraphs). Examples of hereditary graph classes include planar graphs, forests, bipartite graphs, perfect graphs, etc. 
Interestingly, the only graph properties that are known to require PLSs with large certificates 
(e.g. small diameter \cite{censor2020approximate}, non-3-colorability \cite{goos2016locally}, having a non-trivial automorphism \cite{goos2016locally}), are non-hereditary. This raises the question of whether we can improve the result of this paper and design a PLS for the recognition of any geometric intersection graph.


\section{Preliminaries}

Let $G$ be a simple connected $n$-node graph, let $I:V(G)\to \{0,1\}^\star$ be an input function assigning labels to the nodes of $G$, where the size of all inputs is polynomially bounded on $n$. Let $\id:V(G)\to\{1,\dots,\text{poly}(n)\}$ be a one-to-one function assigning identifiers to the nodes. A \emph{distributed language} $\mathcal L$ is a (Turing-decidable) collection of triples $(G,\id,I)$, called \emph{network configurations}.

A distributed interactive protocol consists of a constant series of interactions between a {\emph{prover}} called Merlin, and a {\emph{verifier}} called Arthur. The prover Merlin is centralized, has unlimited computing power and knows the complete configuration $(G,\id,I)$. However, he cannot be trusted. On the other hand, the verifier Arthur is distributed, represented by the nodes in $G$, and has limited knowledge. In fact, at each node $v$, Arthur is initially aware only of his identity $\id(v)$, and his label $I(v)$. He does not know the exact value of $n$, but he knows that there exists a constant $c$ such that $\id(v) \leq n^c$. Therefore, for instance, if one node $v$ wants to communicate  $\id(v)$ to its neighbors, then the message is of size $\cO(\log n)$.

Given any network configuration $(G, \id, I)$, the nodes of $G$ must collectively decide whether $(G, \id, I)$ belongs to some distributed language ${\mathcal L}$. If this is indeed the case, then all nodes must accept; otherwise, at least one node must reject (with certain probabilities, depending on the precise specifications we are considering).

Interactive protocols have two phases: an interactive phase and a verification phase. If Arthur is the party that starts the interactive phase, he picks a random string $r_1$ (known to all nodes of $G$ because we are considering the shared randomness setting) and send it
 to Merlin. Merlin receives $r_1$ and provides every node $v$ with a certificate $c_1(v)$ that is a function of $v$, $r_1$ and $(G,\id,I)$. Then again Arthur picks a random string $r_2$  and sends $r_2$ to Merlin, who, in his turn, provides every node $v$ with a certificate $c_2(v)$ that is a function of $v$, $r_1$, $r_2$ and $(G,\id,I)$. This process continues for a fixed number of rounds. If Merlin is the party that starts the interactive phase, then he provides at the beginning every node $v$ with a certificate $c_0(v)$ that is a function of $v$ and $(G,\id,I)$, and the interactive process continues as explained before.  In the last interaction round, the verification phase begins. This phase is a one-round deterministic algorithm executed at each node. More precisely, every node $v$ broadcasts a message $M_v$ to its neighbors. This message may depend on $\id(v)$, $I(v)$, all random strings generated by Arthur, and all certificates received by $v$ from Merlin. Finally, based on all the knowledge accumulated by $v$ (i.e., its identity, its input label, the generated random strings, the certificates received from Merlin, and all the messages received from its neighbors), the protocol either accepts or rejects at node $v$. Note that Merlin knows the messages each node broadcasts to its neighbors because there is no randomness in this last verification round.

\begin{definition}
Let ${\mathcal V}$ be a verifier and ${\mathcal M}$ a prover of a distributed interactive proof protocol for languages over graphs of $n$ nodes.
If $({\mathcal V}, {\mathcal M})$ corresponds to an Arthur-Merlin $k$-round, $\cO(f(n))$ bandwidth protocol, we write $({\mathcal V}, {\mathcal M}) \in {\dAM}_{\prot}[k,f(n)]$.
\end{definition}

\begin{definition}
Let $\eps \leq 1/3$. The class $\dAM_{\eps}[k,f(n)]$  is the class of languages ${\mathcal L}$ over graphs of $n$ nodes
for which there exists a verifier ${\mathcal V}$ such that, for every configuration $(G,\id,I)$ of size $n$, the
two following conditions are satisfied.

\medskip

\noindent \emph{\completeness}
If $(G,\id,I) \in \cL$ then, there exists a prover $ {\mathcal M}$ such that  $({\mathcal V}, {\mathcal M}) \in {\dAM}_{\prot}[k,f(n)]$  and 
\noindent
\[\mathbf{Pr} \Big{[} \mathcal{V} \mbox{ accepts } (G,\id,I) \mbox{ in every node given } \mathcal{M}\Big{]} \geq 1 - \eps.\]

\medskip

\noindent \emph{\soundness}
 If $(G,\id,I) \notin \cL$ then, for every prover $ {\mathcal M}$ such that  $({\mathcal V}, {\mathcal M}) \in {\dAM}_{\prot}[k,f(n)]$, and
\noindent
\[\mathbf{Pr} \Big{[} \mathcal{V} \mbox{ rejects } (G,\id,I) \mbox{ in at least one nodes given } \mathcal{M}\Big{]} \geq 1-\eps.\]

We also denote $\dAM[k,f(n)]= \dAM_{1/3}[k,f(n)]$, and omit the subindex $\eps$ when its value is obvious from the context. 
\end{definition}
For small values of $k$, instead of writing $\dAM[k,f(n)]$, we alternate \textsf{M}'s and \textsf{A}'s. For instance: $ \dMAM[f(n)] = \dAM[3,f(n)] $. In particular $ \dAM[f(n)]= \dAM[2,f(n)]$. Moreover, we denote $\dM[f(n)]$ the model where only Merlin provides a certificate, and no randomness is allowed (in other words, the model $\dM$ is the PLS model).

In this paper, we are interested mainly in the languages of graphs that are permutation, trapezoid, circle, polygon-circle and unit square. 
Formally,\\

\noindent$\permutationgraph = \{\langle G, \id \rangle \text{ s.t. } G \text{ is a permutation graph}\}.\\
\trapezoid = \{\langle G, \id \rangle \text{ s.t. } G \text{ is a trapezoid graph}\}.\\
\circlegraph = \{\langle G, \id \rangle \text{ s.t. } G \text{ is a cicle graph}\}.\\
\polygoncircle = \{\langle G, \id \rangle \text{ s.t. } G \text{ is a $k$-polygon-circle graph}\}.$\\

We denote by $[n]$ the set $\{0, \dots, n-1\}$ and $S_n$ the set of permutations of $[n]$. In the following, all graphs $G=(V,E)$ are simple and undirected. When the nodes of an $n$-node graph are enumerated with unique values in $[n]$, we denote $G=([n],E)$. In a distributed problem, we always assume that the input graph is connected. We use the standard definitions and notations for (induced) subgraph, neighborhood, path, cycle, tree, clique, etc. For more details we refer to the textbook of Diestel~\cite{diestel2005graph}.


\section{Toolbox} \label{sec:tool}
In our results, we use some previously defined protocols as subroutines.  
In some cases, we consider protocols that solve problems which are more general than just decision problems (as, for instance, the construction of a spanning tree).  
%
\subsection{Spanning Tree and Related Problems} The construction of a spanning tree is an important building block for several protocols in the PLS model. 
Given a network configuration $\langle G, \id \rangle$, the {\spanningtree} problem asks to construct a spanning 
tree $T$ of $G$, where each node has to end up knowing which of its incident edges belong to $T$.  

\begin{proposition}
	\label{prop:spanningtree}
	There is a 1-round protocol for {\spanningtree}  with certificates of size $\cO(\log n)$.
\end{proposition}

From the protocol of Proposition~\ref{prop:spanningtree} it is easy to construct another one  for problem $\sizeofG$, where the nodes, given the input graph $G=(V,E)$,  have to verify the precise value of $|V|$
 (recall that we are assuming that the nodes are only aware of a polynomial upper bound on $n=|V|$). 

\begin{proposition}\cite{korman2010proof}\label{prop:sizeofG}	
	There is a 1-round protocol for $\sizeofG$ with certificates of size $\cO(\log n)$.
\end{proposition}

Finally, for two fixed nodes $s,t \in V$, problem $\stpath$ is defined in the usual way:  given a network configuration $\langle G, \id \rangle$, the output is a path $P$ that goes from $s$ to $t$. 
In other words,  each node must end up knowing whether it belongs to $P$, and, in the positive cases, which of its neighbors are its predecessor and successor in $P$. 

%

\begin{proposition}\cite{korman2010proof}\label{prop:stpath}	
	There is a 1-round protocol for $\stpath$ with certificates of size $\cO(\log n)$.
\end{proposition}

\subsection{Problems Equality and Permutation} 

A second important building block, this time for interactive protocols, is a protocol for solving problem $\equality$, which is defined as follows.
Given $G$ a connected $n$-node graph, each node $v$ receives two natural numbers $a(v)$ and $b(v)$, both of them encoded with $\mathcal{O}(log(n))$ bits.
The problem $\equality$ consists of verifying whether the multi-sets $\mathcal{A} = \{a(v)\}_{v\in V}$ and $\mathcal{B} = \{b(v)\}_{v\in V}$ are equal.

\begin{proposition}\cite{naor2020power} \label{prop:equality}
Problem $\equality$ belongs to {\dAM}$_{1/3}[\log n]$.
\end{proposition}

A closely related problem is \permutation, where some function $\pi$ is given as input, and the nodes must verify whether  $\pi$ is indeed a permutation
(a bijective function from $V$ to $[n]$). Note that the input is given in a distributed way, by given $\pi(v)$ to each node $v\in V$. 
 Using the protocol for $\equality$ as subroutine, it is possible to solve $\permutation$ with  certificates of size $\cO(\log n)$.

\begin{proposition}\cite{naor2020power}\label{prop:permu}
Problem $\permutation$ belongs to {\dMAM}$_{1/3}[\log n]$.
\end{proposition}


We now introduce a new problem called  $\COP$, which is defined for inputs of the form $\langle G = (V,E), \id, (x,\pi) \rangle$,
where the nodes must verify that: (i) $\pi$ is a bijection  from $V$ to $[n]$; (ii) $x$ is an injective function from $V$ to $[N]$, where $N \geq n$; and
(iii) for every $u,v \in V$, $\pi(u) \geq \pi(v)$ $\iff$ $x(u) \geq x(v)$.

%
%

\begin{proposition}\label{prop:COP}
Problem $\COP$ belongs to ${\dMAM}[\log N]$.
\end{proposition}

\begin{proof}
The following is a protocol for $\COP$. In the first round, each node $v$ receives from the prover:
\begin{itemize}
\item The certification for the size of $V$ and that $\pi$ is an injective function. 
\item $a(v) = (x(v),\pi(v))$ and $b(v) = (y(v),\pi(v)+1 \mod n)$.
\end{itemize}

Suppose that $\pi$ is an injective function and that $n$ is known by all the nodes. Observe that, if $\{a(v)\}_{v\in V}$ and $\{b(v)\}_{v\in V}$ are equal, then $y(v) = x(u)$, where $u$ is the successor of $v$, i.e. $\pi(u) = \pi(v)+1$. Then,  $\langle G = (V,E), \id, (x,\pi) \rangle$  is a yes-instance of $\COP$ if and only if $\pi$ is an injective function, $\{a(v)\}_{v\in V}=\{b(v)\}_{v\in V}$ , and $x(v) \leq y(v)$ for each node $v$ such that $\pi(v)< n-1$.

Then, in the two remaining rounds, the nodes interact with the prover in order to prove that $\pi$ is an injective function and that $\{a(v)\}_{v\in V}$, $\{b(v)\}_{v\in V}$ are equal multi-sets, using the protocols for $\permutation$ and $\equality$, respectively. The nodes also check that $x(v) \leq y(v)$ (except for the node $v$ such that $\pi(v) = n-1$). The communication bounds, as well as the correctness and soundness of the protocol follows from the ones of protocols for $\spanningtree$, $\sizeofG$, $\equality$ and $\permutation$.
\qed
\end{proof}

Note that our protocol for $\COP$  can be easily extended to the case when  the range of function $x$ is a set $S$ of size $N$ that admits a total order.  


\section{Permutation Graphs} \label{sec:permutation}

We begin this section with a characterization of permutation graphs. Then, we use this characterization to prove the existence of a PLS solving \permutationgraph. Let us call $S_n$ the set of permutations of $[n$].  Given $\pi \in S_n$, we say that a pair $i,j \in [n]$ is an \emph{inversion under $\pi$}, or analogously, a $\pi$-inversion, if $$(i-j)(\pi(i)-\pi(j))<0.$$ The definition of a permutation graph can be restated as follows (see \cite{Brandstdt1999}). A graph $G=([n],E)$ is a permutation graph if there exists a permutation $\pi \in S_n$ such that  for every $i,j \in [n]$, the pair $\{i,j\}$ is an edge of $G$, if and only $i,j$ is an inversion under $\pi$. In such a case we say that $\pi$ is a \emph{proper permutation model} of $G$. 

Let us fix a graph $G = ([n],E)$. We say that a permutation $\pi$ is a \emph{semi-proper permutation model} of $G$ if $\pi$ satisfies that for every pair $i, j \in [n]$, if $\{i,j\}$ is an edge of $G$ then $i,j$ is an inversion under $\pi$.The following remark and lemma characterize permutation graphs.

\begin{remark}\label{rem:pyq} Let $G = ([n],E)$ be a permutation graph, $\pi$ permutation model and let $\{i, j\}\in E$ an edge such that $i<j$.  Observe that, since the $G$ is connected, then all integer values $s \in [\pi(i)+1,\pi(j)-1]$ satisfy that $s = \pi(k)$, for some node $k$ that is neighbor of $i$ or $j$. Similarly, all integer values $h \in [i +1,j-1]$ are neighbors of $i$ or $j$. 
\end{remark}

 We denote $N^+_G(i)$ and $N^-_G(i)$ the sets of neighbors of node $i$ with, respectively, higher and fewer identifiers than $i$. Formally,  \(N^+_G(i) = \{j\in N_G(i): j>i\}\), and \(N^-_G(i) = \{j\in N_G(i): j<i\}\). We also denote by $d^+(i)$ and $d^-(i)$ the number of neighbors of $i$ with, respectively, higher and lower identifiers than $i$. Formally, $d^+(i) = |N^+_G(i)|$ and $d^-(i) = |N^-_G(i)|$. From previous definition is direct that $|N_G(i)| = d^+(i) + d^-(i)$. 

Remember that we are only considering undirected graphs, so the notation $N^+(i)$, $N^-(i)$, $d^+(i)$ and $d^-(i)$ must not be confused with the in/out neighborhood or degrees used for directed graphs.



\begin{lemma}\label{lem:permcaract}
Let $G=([n],E)$ be a graph and let $\pi$ be a semi-proper permutation model of~$G$. Then $\pi$ is a proper permutation model of $G$ if and only if, for every $i\in[n]$, $$i + d^+(i) = \pi(i) + d^-(i).$$
\end{lemma}

\begin{proof}
Suppose first that $G$ is a permutation graph and $\pi$ is a proper permutation model for $G$. Given a node $i\in [n]$, observe that all neighbors $j\in N^+(i)$  satisfy $0\leq\pi(j)< \pi(i)$. Otherwise, $(i,j)$ would not be an inversion under $\pi$. Analogously, if $j\in N^+(i)$ then $\pi(i)<\pi(j)\leq n$. Hence, the pre-images of $[\pi(i)]$ under $\pi$ are the nodes in $N^+(i)$ and nodes in $[i]\setminus N^-(i)$. Therefore, $\pi(i) = d^+(i) + i  - d^-(i)$, from which the equality $i + d^+(i) = \pi(i) + d^-(i)$ is deduced.

Let us suppose that $G$ is not a permutation graph. We show the existence of a node $i^* \in [n]$ such that  $i^* + d^+(i^*) \neq \pi(i^*) + d^-(i^*)$. Remember that we are assuming that $\pi$ is a semi-proper permutation model for $G$. Then, we have that necessarily there exists a pair $\{i,j\}\not\in E$ such that $i,j$ is an inversion under $\pi$.  For a node $i\in [n]$, let us denote by $a^-(i)$ and $a^+(i)$ the number of nodes with, respectively, fewer and larger  identifier than $i$, forming an inversion with $i$, but that are not adjacent to $i$. Formally:
\[a^+(i) = |\{j\in \{i+1, \dots, n\}: \{i,j\}\notin E  \textrm{ and }  \{i,j\} \textrm{ is a } \pi\text{-}inversion \}| \]
\[a^-(i) = |\{j\in \{1, \dots, i-1\}: \{i,j\}\notin E  \textrm{ and }  \{i,j\} \textrm{ is a } \pi\text{-}inversion \}| \]

The pre-images of the set $[\pi(i)]$ under $\pi$ are: The $d^+(i)$ neighbors higher than~$i$ (because the nodes in $N^{+}(i)$ form an inversion with $i$, as $\pi$ is semi-proper); the $a^+(i)$ nodes that are not neighbors of $i$ but form an inversion with $i$; and the nodes in set $[i]$ that do not form a inversion with $i$, which are exactly $i-(d^-(i) + n^-(i))$, then we have that
\(\pi(i) =  d^+(i) + a^+(i) + i - (a^-(i) + d^-(i)), \)
from which it is concluded that 
\[ i + d^+(i) \not =  \pi(i)  + d^-(i) \iff a^-(i) \neq a^+(i). \]

Therefore, we have to show that there exists a node $i$ such that $a^-(i)\not = a^+(i)$. Let $U\subseteq V$ be the set  of nodes forming an inversion with another node that is not its neighbor. Formally
\[U = \{j\in [n] \colon \exists k\in \{j,k\}\notin E \textrm{ and } \{j,k\} \textrm{ is a } \pi\text{-}inversion\}\]

Let $i^* = \min U$. Then, by definition of $U$ there exists a $k\in [n]$ such that $\{i^*,k\}$ is an inversion under $\pi$ and $\{i^*,k\}\notin E$. Also, since $i^*$ is the minimum node in $U$, necessarily $k>i^*$. Hence $a^-(i^*) = 0$ and $a^+(i^*)\geq 1$. Therefore $i^*$ satisfies the condition $n^-(i)\not = n^+(i)$ and therefore \(i^* + d^+(i^*) \neq  \pi(i^*) + d^-(i^*)\).
\end{proof}

We are now ready to define our protocol and main result regarding \permutationgraph.

\begin{theorem}\label{theo:permutaitongraph}
There is a 1-round proof labelling scheme for \permutationgraph\ with certificates of size $\cO(\log n)$.
\end{theorem}

\begin{proof}
 The next is a one round protocol for \permutationgraph: The certificate provided by the prover to node $v$ is interpreted as follows: 
	
		\begin{enumerate}
			\item The certification of the number of nodes $n$, according to a protocol for $\sizeofG$.
			\item  Values $\ell_1(v), \ell_2(v) \in [n]$, where $\ell_1$ and $\ell_2$ are injective functions from $V$ to $[n]$. The pair $(\ell_1(v), \ell_2(v))$ is interpreted as a value of a permutation $\pi$ such that $\pi(\ell_1(v)) = \ell_2(v)$. 
			\item Value $p_v$ corresponding to the minimum value grater than $\ell_1(v)$ that not an image under $\ell_1$ of a neighbor of $v$. Formally, 
			 \(p_v = \min\{ k \in \{\ell_1(v)+1, \dots, n-1\} :  \forall u \in N(v), k \neq \ell_1(u) \}\).
			\item  Value $q_v$ corresponding to the minimum value  greater that $\ell_2(v)$ that is not an image under $\ell_2$ of a neighbor of $v$. Formally,
			 \(q_v = \min\{ k \in \{\ell_2(v)+1, \dots, n-1\} :  \forall u \in N(v), k \neq \ell_2(u) \}\).
			
			\item The certification of a path $P_1$ between the nodes $u$ and $w$ such that $\ell_1(u) = 0$ and $\ell_1(w) =n-1$, according to the protocol for $\stpath$.
			\item The certification of a path $P_2$ between the nodes $u$ and $w$ such that $\ell_2(u) = 0$ and $\ell_2(w) =n-1$, according to the protocol for $\stpath$.
					\end{enumerate}
		
Then, in the verification round, each node shares with its neighbors their certificates. Using that information each node $v$ can compute $d^+(v)$ and $d^-(v)$, and check the following conditions:

\begin{enumerate}[label=\lipics{\alph*.}]
	\item The correctness of the number of nodes $n$ according to the protocol of $\sizeofG$.
	\item The correctness of the paths $P_\pi$ and $P_\ell$ according to the protocol of $\stpath$.
	\item The values $\ell_1(v),\ell_2(v)$ belong to the set $[n]$.
	\item The inversion with its neighbors, i.e., $\forall \omega\in N(v), (\ell(v)-\ell(\omega))(\pi(v)-\pi(\omega))< 0$.
	\item The equality $\ell_1(v) + d^+(v) = \ell_2(v) + d^-(v)$ holds.
	\item For each $\omega\in N(v)$ such that $\ell_1(\omega)<p_\omega<\ell_1(v)$,  $p_\omega = \ell_1(u)$ for some $u\in N(v)$.
	\item For each $\omega\in N(v)$ such that $\ell_2(\omega)<q_\omega<\ell_2(v)$,  $q_\omega = \ell_2(u)$ for some $u\in N(v)$.

\end{enumerate}

We now analyze the soundness and completeness of our protocol.  \\

\noindent \textbf{Completeness:} Suppose that $G$ is a permutation graph. Then a honest prover will choose $\ell_1$ and $\ell_2$ such that $\pi: [n] \rightarrow [n]$ defined by  $\pi(\ell_1(v)) = \ell_2(v)$ is a proper permutation model of $G$. If the prover sends the real value of $n$, the nodes will verify condition \lipics{a} according to the completeness of the protocol for $\sizeofG$. Similarly, if the paths $P_1$ and $P_2$ are valid, condition \lipics{b} is verified according to the completeness of the protocol for $\stpath$. Once that condition \lipics{a} is verified, then \lipics{c}, \lipics{d} and \lipics{e} can be verified looking to the certificates in the neighborhood. Finally, the correctness of $p_v$ and $q_v$ are verified by conditions \lipics{f} and \lipics{g}, which are satisfied by \cref{rem:pyq}.\\

    
    \noindent \textbf{Soundness:} Suppose $G$ is not a permutation graph. If a dishonest prover provides a false value of $n$, or false paths $P_1$ or $P_2$, then at least one node will reject it by correctness of the protocols for $\sizeofG$ and $\stpath$, respectively. Then, we can assume that the prover has not cheated on those values. 
    
    Suppose that the prover gives a function $\ell_1$ such that $\{\ell_1(v)\}_{v\in V}\neq [n]$.  If some $\ell_1(v)$ is not in the set $[n]$ then $v$ will reject in the verification of condition \textbf{b}. Then, we assume the existence of $j\in[n]$ such that $j\neq \ell_1(v)$ for all $v\in V$. As Merlin cannot send an invalid path $P_1$, necessarily $0<j<n-1$. Also, by correctness of the path, there exists nodes $u_1,u_2$ in the path such that $\{u,v\}\in E$ and $\ell_1(u)<j<\ell_1(v)$. From all possible choices of $u$ and $v$, let us choose the pair such that $\ell_1(u)$ is maximum. Now we prove that $v$ fails to check condition \textbf{f} and rejects. Indeed, as no node has $j$ as image through $\ell_1$, then $p_{u}\leq j$.  If $p_{u} = j$, then $v$ fails to check condition \textbf{f} and rejects. Suppose that $p_u < j$ and $v$ verifies condition \textbf{f}. Then there must exist $\omega \in N(v)$ such that $p_u = \ell_1(\omega) < j$, contradicting the choice of $u$. We deduce that if condition \textbf{f} is verified by every node, necessarily $\ell_1$ is an injective function from $V$ into $[n]$.
    
   By an analogous argument, we deduce that if condition \textbf{g} is verified by every node, then necessarily $\ell_2$ is an injective function from $V$ into $[n]$.
   
 Suppose then that the dishonest prover provides the right value of $n$, as well as injective functions $\ell_1$ and $\ell_2$.  If condition \textbf{d} is verified, then $\pi$ is a semi-proper permutation model of $G$. Then, since $G$ is not a permutation graph, at least has to fail to verify condition \textbf{e} by \cref{lem:permcaract}. 
 
We now analyze the communication complexity of the protocol:  the certification for $\sizeofG$ and $\stpath$ is $\cO(\log n)$, given by  \cref{prop:sizeofG} and \cref{prop:stpath}. On the other hand, for each $v\in V$, the values $\ell_1(v)$, $\ell_2(v)$, $q_v$, $p_v$ are $\cO(\log n)$ as they are numbers in $[n]$. Therefore, overall the total communication is $\cO(\log n)$. 
\end{proof}


\section{Trapezoid Graphs} \label{sec:trapezoid}

The protocol for Trapezoid is a sort of generalization of the protocol for permutation graphs. Indeed, for this class we can give an analogous characterizations, that later are used to build a compact one-round proof labeling scheme for $\trapezoid$.

Remember that in a model of a trapezoid graph, there are two parallel lines $\mathcal{L}_t$ and $\mathcal{L}_b$. We denote this lines the \emph{top and bottom lines}, respectively. End each trapezoid has sides contained in each line, and then  defined by four vertices, two in the top line, and two in the bottom line. Formally, each trapezoid $T$ is defined by the set $T = \{t_1,t_2, b_1, b_2\}$, where $t_1 < t_2$ and $b_1 < b_2$,  with $t_1, t_2 \in \mathcal{L}_t$ and $b_1, b_2 \in \mathcal{L}_b$ (see \Cref{fig:trapezoidvertex}).  

	\begin{figure}[h!]
		\centering 
		\tikzset{every picture/.style={line width=0.75pt}} 

\begin{tikzpicture}[x=0.75pt,y=0.75pt,yscale=-1,xscale=1]

\draw  [color={rgb, 255:red, 208; green, 2; blue, 27 }  ,draw opacity=1 ][fill={rgb, 255:red, 208; green, 2; blue, 27 }  ,fill opacity=0.2 ] (350,45) -- (310,125) -- (250,125) -- (230,45) -- cycle ;
\draw [line width=0.75]    (220,45) -- (360,45) ;
\draw [line width=0.75]    (220,125) -- (360,125) ;
\draw [line width=0.75]    (310,120) -- (310,130) ;
\draw [line width=0.75]    (250,120) -- (250,130) ;
\draw [line width=0.75]    (350,40) -- (350,50) ;
\draw [line width=0.75]    (230,40) -- (230,50) ;
\draw [line width=0.75]  [dash pattern={on 0.84pt off 2.51pt}]  (360,125) -- (375,125) ;
\draw [line width=0.75]  [dash pattern={on 0.84pt off 2.51pt}]  (360,45) -- (375,45) ;
\draw [line width=0.75]  [dash pattern={on 0.84pt off 2.51pt}]  (205,125) -- (220,125) ;
\draw [line width=0.75]  [dash pattern={on 0.84pt off 2.51pt}]  (205,45) -- (220,45) ;

\draw (281,75.4) node [anchor=north west][inner sep=0.75pt]    {$T$};
\draw (246,132.4) node [anchor=north west][inner sep=0.75pt]  [font=\footnotesize]  {$b_{1}$};
\draw (306,132.4) node [anchor=north west][inner sep=0.75pt]  [font=\footnotesize]  {$b_{2}$};
\draw (226,25.4) node [anchor=north west][inner sep=0.75pt]  [font=\footnotesize]  {$t_{1}$};
\draw (347,25.4) node [anchor=north west][inner sep=0.75pt]  [font=\footnotesize]  {$t_{2}$};

\end{tikzpicture}
		\caption{Each trapezoid $T$ is defined by the set $T = \{ b_1, b_2, t_1, t_2\}$.}
		\label{fig:trapezoidvertex}
	\end{figure}
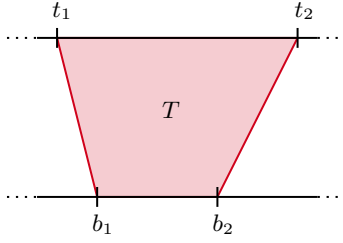

The definition of a trapezoid graph can be restated as follows (see \cite{Brandstdt1999}): A trapezoid graph $G = (V,E)$ is the intersection graph of a set of trapezoids $\{T_v\}_{v\in V}$ satisfying the following conditions. The vertices of each trapezoid have values in $[2n]$,  two corresponding to  the upper line and the other to the bottom line. The vertices defining the set $\{T_v\}_{v\in V}$, are all different, i.e., no pair of trapezoids share vertices. Therefore, in both the top and the bottom lines, each element in $[2n]$ correspond to a vertex of some trapezoid. The trapezoid model in the example of ref{fig:Extrapezoid} satisfies these conditions. 

 For $v \in V$, we call $\{t_1(v), t_2(v), b_1(v), b_2(v)\}$ the vertices of $T_v$.  Moreover, we say that $\{t_1(v), t_2(v), b_1(v), b_2(v)\}$ are the \emph{vertices} of node $v$. In the following, a trapezoid model satisfying the abode conditions is called a \emph{proper trapezoid model} for $G$. Given a graph $G=(V,E)$ (that is not necessarily a trapezoid graph), a \emph{semi-proper trapezoid model} for  $G$ is a set of trapezoids $\{T_v\}_{v\in V}$ satisfying previous conditions, such that, for every $\{u,v\} \in E$, the trapezoids $T_v$ and $T_u$ have nonempty intersection. The difference between a proper and a semi-proper model is that in the first we also ask every pair of non-adjacent edges have non-intersecting trapezoids.  

	Given a trapezoid graph $G = (V,E)$ and a proper trapezoid model $\{T_v\}_{v\in V}$, we define the following sets for each $v\in V$:
	\begin{align*}
		F_t(v) & = \{i\in [2n]\mid i<t_1(v)\text{ and }i \in \{t_1(w), t_2(w)\} \text{ for some }w\notin N(v)\} \\
		F_b(v) & = \{i\in [2n]\mid i<b_1(v)\text{ and }i \in \{b_1(w), b_2(w)\} \text{ for some }w\notin N(v)\}
	\end{align*}

	We also call $f_t(v)  = |F_t(v)|$ and $f_b(v) = |F_b(v)|$. The following lemmas characterize  trapezoid graphs.

\begin{lemma}\label{lem:trapcaract}
Let $G=(V,E)$ an connected trapezoid a graph with $n$ nodes. Then each proper trapezoid model  $\{T_v\}_{v\in V}$  of $G$ satisfies for every $v\in V$: $$b_1(v)  -  f_b(v) = t_1(v)  - f_t(v)$$
\end{lemma}

\begin{proof}
Let $\{T_v\}_{v\in V}$ be a proper trapezoid model of $G$.  Then, given a node $v\in V$, all the coordinates in $F_t(v)$ are vertices of some $w\neq N(v)$. Such trapezoids $T_w$ have their two upper vertices in the set $[t_1(v)]$ and their two lower vertices in $[b_1(v)]$, as otherwise $T_w$ and $T_v$ would intersect. Then, the cardinality of the set $[t_1(v)] \setminus F_t(v)$ is even, and the same holds for $[b_1(v)] \setminus F_b(v)$. Moreover, the cardinality of the set $[t_1(v)]  \setminus F_t(v)$ equals the cardinality of $[b_1(v)]  \setminus F_b(v)$, as every position in $[2n]$ corresponds to a vertex of some trapezoid's node. We deduce that \[t_1(v) - f_t(v) = |\{1, \dots, t_1(v)\}  \setminus F_t(v)| = |\{1, \dots, b_1(v)\}  \setminus F_b(v)| = b_1(v) - f_b(v).\]
\end{proof}

\begin{lemma}\label{lem:nontrapcarac}
Let $G = (V,E)$ be a $n$-node graph that is not a trapezoid graph. Then, for every semi-proper trapezoid model $\{T_v\}_{v\in V}$ of $G$, at least one of the following conditions is true: 

\begin{enumerate}
    \item $\exists v\in V$ such that some value in $\{b_1(v), \dots, b_2(v)\}$ or $\{t_1(v),\dots, t_2(v)\}$ is a vertex of $\omega\notin N(v)$.
    \item  $\exists v\in V$ such that $b_1(v)  -  f_b(v) \neq t_1(v)  - f_t(v)$.
\end{enumerate}

\end{lemma}

\begin{proof}

Let $G$ be a graph that is not a trapezoid graph and $\{T_v\}_{v\in V}$ a semi-proper trapezoid model. As $G$ is not a permutation graph, by definition necessarily there exist a pair $\{v,\omega\}\not\in E$ such that $T_v \cap T_\omega \neq \emptyset$. We distinguish two possible cases (see \Cref{fig:nontrapcarac}): 
\begin{enumerate}
\item  $[b_1(v), b_2(v)]_\N \cap [b_1(\omega), b_2(\omega)]_\N\neq\emptyset$ or $[t_1(v) , t_2(v)]_\N\cap [t_1(\omega), t_2(\omega)]_\N\neq\emptyset$.
\item $[b_1(v), b_2(v)]_\N \cap [b_1(\omega), b_2(\omega)]_\N=\emptyset$ and $[t_1(v) , t_2(v)]_\N\cap [t_1(\omega), t_2(\omega)]_\N=\emptyset$.
\end{enumerate}

	\begin{figure}[h!]
		\centering 
		\tikzset{every picture/.style={line width=0.75pt}} 

\begin{tikzpicture}[x=0.75pt,y=0.75pt,yscale=-1,xscale=1]

\draw  [color={rgb, 255:red, 74; green, 144; blue, 226 }  ,draw opacity=1 ][fill={rgb, 255:red, 74; green, 144; blue, 226 }  ,fill opacity=0.2 ] (375,60) -- (475,140) -- (395,140) -- (335,60) -- cycle ;
\draw  [color={rgb, 255:red, 208; green, 2; blue, 27 }  ,draw opacity=1 ][fill={rgb, 255:red, 208; green, 2; blue, 27 }  ,fill opacity=0.2 ] (455,60) -- (375,140) -- (335,140) -- (395,60) -- cycle ;
\draw [line width=0.75]    (315,60) -- (475,60) ;
\draw [line width=0.75]    (315,140) -- (475,140) ;
\draw [line width=0.75]    (375,135) -- (375,145) ;
\draw [line width=0.75]    (395,135) -- (395,145) ;
\draw [line width=0.75]    (415,135) -- (415,145) ;
\draw [line width=0.75]    (435,135) -- (435,145) ;
\draw [line width=0.75]    (455,135) -- (455,145) ;
\draw [line width=0.75]    (475,135) -- (475,145) ;
\draw [line width=0.75]    (315,135) -- (315,145) ;
\draw [line width=0.75]    (335,135) -- (335,145) ;
\draw [line width=0.75]    (355,135) -- (355,145) ;
\draw [line width=0.75]    (375,55) -- (375,65) ;
\draw [line width=0.75]    (395,55) -- (395,65) ;
\draw [line width=0.75]    (415,55) -- (415,65) ;
\draw [line width=0.75]    (435,55) -- (435,65) ;
\draw [line width=0.75]    (455,55) -- (455,65) ;
\draw [line width=0.75]    (475,55) -- (475,65) ;
\draw [line width=0.75]    (315,55) -- (315,65) ;
\draw [line width=0.75]    (335,55) -- (335,65) ;
\draw [line width=0.75]    (355,55) -- (355,65) ;
\draw  [color={rgb, 255:red, 74; green, 144; blue, 226 }  ,draw opacity=1 ][fill={rgb, 255:red, 74; green, 144; blue, 226 }  ,fill opacity=0.2 ] (255,60) -- (235,140) -- (175,140) -- (195,60) -- cycle ;
\draw  [color={rgb, 255:red, 208; green, 2; blue, 27 }  ,draw opacity=1 ][fill={rgb, 255:red, 208; green, 2; blue, 27 }  ,fill opacity=0.2 ] (215,60) -- (215,140) -- (155,140) -- (135,60) -- cycle ;
\draw [line width=0.75]    (115,60) -- (275,60) ;
\draw [line width=0.75]    (115,140) -- (275,140) ;
\draw [line width=0.75]    (175,135) -- (175,145) ;
\draw [line width=0.75]    (195,135) -- (195,145) ;
\draw [line width=0.75]    (215,135) -- (215,145) ;
\draw [line width=0.75]    (235,135) -- (235,145) ;
\draw [line width=0.75]    (255,135) -- (255,145) ;
\draw [line width=0.75]    (275,135) -- (275,145) ;
\draw [line width=0.75]    (115,135) -- (115,145) ;
\draw [line width=0.75]    (135,135) -- (135,145) ;
\draw [line width=0.75]    (155,135) -- (155,145) ;
\draw [line width=0.75]    (175,55) -- (175,65) ;
\draw [line width=0.75]    (195,55) -- (195,65) ;
\draw [line width=0.75]    (215,55) -- (215,65) ;
\draw [line width=0.75]    (235,55) -- (235,65) ;
\draw [line width=0.75]    (255,55) -- (255,65) ;
\draw [line width=0.75]    (275,55) -- (275,65) ;
\draw [line width=0.75]    (115,55) -- (115,65) ;
\draw [line width=0.75]    (135,55) -- (135,65) ;
\draw [line width=0.75]    (155,55) -- (155,65) ;

\end{tikzpicture}
		\caption{A representation of the two possible cases. In the first case, depicted in left, at least one vertex of a trapezoid is contained in the other. In the second case, in the right hand, the trapezoids intersect, but not in the vertices. }
		\label{fig:nontrapcarac}
	\end{figure}
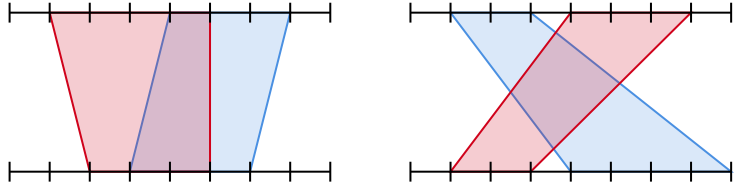

Clearly if the first case holds, then condition \lipics{1} is satisfied. Suppose then that there is no pair $\{v,\omega\}\not\in E$ such that $T_v \cap T_\omega \neq \emptyset$ satisfying the first case. Then necessarily the second case holds. Let $u$ be a node for which  exists $\omega \in V\setminus N(u)$ such that $T_u \cap T_w \neq \emptyset$. For all possible choices of $u$, let us pick the one such that $b_1(u)$ is minimum. Then $u$ satisfies the following conditions:
\begin{itemize}
	\item[\lipics{(a)}] Exists a node $\omega\in V$ such that $\omega \notin N(v)$ and $T_{u}\cap T_\omega\neq\emptyset$
	\item[\lipics{(b)}] All nodes $\omega\in V$ such that $\omega \notin N(v)$ and $T_{u}\cap T_\omega\neq\emptyset$ satisfy that $t_2(\omega) < t_1(u)$ and $b_2(u)<b_1(\omega) $
	\item[\lipics{(c)}] None of the positions in $\{1,\dots, b_1(u)\}$ is occupied by a vertex of a node $\omega$ such that $\{u,\omega\}\notin E$ and $T_u \cap T_\omega  \neq \emptyset$. 
\end{itemize}

Observe that conditions \lipics{(a)} and \lipics{(b)} imply that $t_1(u) - f_t(u) > 0$, while condition \lipics{(c)} implies that $b_1(u) - f_b(u) = 0$. We deduce that condition \lipics{2} holds by $u$.
\end{proof}

We are now ready to define our protocol and main result regarding \trapezoid.

\begin{theorem}\label{theo:trapezoid}
There is a 1-round proof labelling scheme for \trapezoid\ with certificates of size $\cO(\log n)$.
\end{theorem}

\begin{proof}

The following is a one-round proof labeling scheme for \trapezoid\:\\

Given an instance $\langle G = (V,E), \id \rangle$, the certificate provided by the prover to node $v \in V$ is interpreted as follows.
		\begin{enumerate}
			\item  The certification of the total number of nodes $n$, according to some protocol for $\sizeofG$.
			\item  Values $b_1(v), b_2(v) , t_1(v), t_2(v)\in [2n]$,  such that $b_1(v) < b_2(v)$ and $t_1(v) < t_2(v)$, representing the vertices of a trapezoid $T_v$.
			\item  Value $p_v$ corresponding to the minimum position in the upper line greater that $t_1(v)$ that is not a vertex of a neighbor of $v$.
			\item Value $q_v$ corresponding minimum position in the lower line grater than $b_1(v)$ that is not a vertex of a neighbor of $v$.
			\item The certification of a path $P_t$  between the node with vertice $0$ and the node with vertice $2n-1$ in the upper line (respect assignment in $2.$) and a path $P_b$ between the node node with vertice $0$ and the node with vertice $2n-1$ in the lower line. Both paths according to a protocol for $\stpath$.

		\end{enumerate}

Then, in the verification round, each node shares with its neighbors their certificates. Using that information each node $v$ can compute $f_t(v)$ and $f_b(v)$, and check the following conditions: 

\begin{enumerate}[label=\lipics{\alph*.}]
	\item The correctness of the value of $n$, according to some protocol for $\sizeofG$.
	\item The correctness of the paths  $P_b$ and $P_t$, according to a protocol for $\stpath$.
	\item The vertices of the trapezoid of $v$ are in $[2n]$.
	\item $T_v\cap T_\omega\neq\emptyset$ for all $\omega\in N(v)$. 
	\item All values in $\{t_1(v)+1,\dots, t_2(v)-1\}$ and $\{b_1(v)+1,\dots, b_2(v)-1\}$ are vertex of some neighbor of $v$.

	\item $t_2(v)<p_v$ and $b_2(v)< q_v$.
	\item If $\omega\in N(v)$ and $p_\omega<t_2(v)$, then $v$ verifies that $p_\omega$ is a vertex of some other neighbor.
	\item If $\omega\in N(v)$ and $q_\omega<b_2(v)$, then $v$ verifies that $q_\omega$ is a vertex of some other neighbor.
	\item $b_1(v) - f_b(v) = t_1(v) - f_t(v)$.

\end{enumerate}

We now analyze the soundness and completeness of our protocol.  \\

\noindent\textbf{Completeness:} Suppose that $G$ is a trapezoid graph. An honest prover just has to send the real number of nodes $n$, a trapezoid model $\{T_v\}_{v\in V}$ of G and valid paths $P_b$ and $P_t$ according the trapezoid model. Then, the nodes will verify \lipics{a}, \lipics{b} by the completeness of the protocols for $\sizeofG$ and $\stpath$. Conditions \lipics{c}, \lipics{d}, \lipics{e} ,\lipics{f}, \lipics{g} and \lipics{h} are verified by the correctness of the model  $\{T_v\}_{v\in V}$. Condition \lipics{i} is also verified, by Lemma \ref{lem:trapcaract}.\\
    
\noindent \textbf{Soundness:} Suppose $G$ is not a trapezoid graph. If a dishonest prover provides a wrong value of $n$, or wrong paths $P_t$ or $P_b$, then at least one node will reject verifying \lipics{a} or \lipics{b}. Then, we assume that the prover cannot cheat on these values.  
    
    Suppose that the prover gives values $\{T_v\}_{v\in V}$ such that is  fulfilled $\bigcup_{v\in V}\{t_1(v), t_2(v)\}\neq [2n]$. If some vertex of a node is not in the set $[2n]$, then that node fails to verify condition \lipics{c} and rejects.  Without loss of generality, we can assume that there exists a $j \in [2n]$ such that $t_1(v), t_2(v) \neq j$, for every $v\in V$. If a node $\omega$ satisfies that $t_1(\omega) < j < t_2(\omega)$, then node $\omega$ fails to verify condition \lipics{e} and rejects. Then $j$ is not contained in any trapezoid. As $P_t$ is correct, $j$ must be different than $1$ and $2n$. Also by the correctness of $P_t$, there exist a pair of adjacent nodes $u,v \in V$ such that $t_2(u) < j < t_1(v)$.  From all possible choices for $u$ and $v$, we pick the one such that $t_2(u)$ is maximum. We claim that $v$ fails to check condition \lipics{g}. Since $j$ is not a vertex of any node, then $p_u \leq j$. If $v$ verifies condition \lipics{g}, then necessarily $p_u < j$. Then, there must exist a node $\omega \in N(v)$ such that $p_u = t_1(\omega)$. But since we are assuming that $j$ is not contained in any trapezoid, we have that $t_2(\omega) < j$, contradicting the choice of $u$. 
  
    Therefore, if conditions  \lipics{a} - \lipics{h} are verified, we can assume that the nodes are given a semi-proper trapezoid model of $G$. Since we are assuming that $G$ is not a trapezoid graph, by \Cref{lem:nontrapcarac} we deduce that condition \lipics{i} cannot be satisfied and some node rejects.     
  
  We now analyze the communication complexity of the protocol:  the certification for $\sizeofG$ and $\stpath$ is $\cO(\log n)$, given by  ref{prop:sizeofG} and ref{prop:stpath}. On the other hand, for each $v\in V$, the values $b_1(v)$, $b_2(v)$, $t_1(v)$, $t_2(v)$, $p_v$, $q_v$ are computable in $\cO(\log n)$ space as they are numbers in $[2n]$. Overall the total communication is $\cO(\log n)$. 
    \end{proof}

\section{Circle Graphs} \label{sec:circle}
We now tackle the problem $\circlegraph$. Unlike the cases of permutation and trapezoid graphs, for this class we are unable to provide a one-round protocol. Instead, we give a compact three round interactive protocol, which is based in the following characterization of cicle graphs  (see \cite{Brandstdt1999}):

A graph $G = (V,E)$ is a circle graph if and only each node $v\in V$ can be associated to an interval $I_v = [m_v, M_v] \subseteq \mathbb{R}$ such that:
\begin{itemize}
\item[\lipics{C1}] For each $u,v \in V$, $\{u,v\} \in E$ then $I_u\cap I_v \neq \emptyset$,  $I_u\not\subseteq I_v$ and $I_v\not\subseteq I_u$.
\item[\lipics{C2}]  For each $u,v \in V$, $I_u\cap I_v \neq \emptyset$,  $I_u\not\subseteq I_v$ and $I_v\not\subseteq I_u$ then $\{u,v\} \in E$.
\item[\lipics{C3}]  For each $v\in V$,  $m_v, M_v \in [2n]$, and $\{m_v\}_{v\in V} \cup \{M_v\}_{v\in V} = [2n]$, i.e., for every pair of different nodes $u,v \in V$, $m_u \neq m_v$ and $M_u \neq M_v$.

\end{itemize}


In words, for every circle graph of size $n$ there is a collection of $n$ intervals with extremes in $[2n]$ whose extremes do not coincide, and where two nodes are adjacent if and only if their corresponding intervals have nonempty intersection, and one interval is not included in the other. We remark that the later characterization is not an intersection model of a circle graphs, as we ask more than simply the intersection of the objects. 
Given a graph $G=(V,E)$, a set of intervals $\{I_v\}_{v\in V}$ satisfying conditions \lipics{C1}, \lipics{C2} and \lipics{C3} is called a \emph{proper model} for $G$. A set of intervals satisfying conditions \lipics{C1} and \lipics{C3} is called a \emph{semi-proper model} for $G$.

Let $G = (V,E)$ be a graph, $\{[m_v, M_v]\}_{v \in V}$ be a semi-proper model for $G$, and let $v \in V$. We denote by $n_m(v)$ the number of nodes $u\in N(v)$ whose lower limit $m_u$ is such that  $m_v \leq m_u \leq M_v$. Similarly, we denote by $n_M(v)$ the number of nodes $u \in N(v)$ such that $m_v \leq M_u \leq M_v$. We also denote by $\pi_m(v)$ (respectively $\pi_M(v)$) the total number of nodes $u\in V$ such that $m_u < m_v$ (respectively $M_u < M_v$). 

\begin{lemma}\label{lem:circchar}
Let $G=(V,E)$ be a graph and let $\{[m_v, M_v]\}_{v \in V}$ be a semi-proper model for $G$. Then    $\{[m_v, M_v]\}_{v \in V}$  is proper if and only if, for every $v\in V$, $$2(\pi_M(v) + \pi_m(v)) = M_v + m_v + n_M(v) - n_m(v)$$
\end{lemma}

\begin{proof}
  Let $v\in V$ be an arbitrary node. Observe that there are $n - \pi_M(v)$ nodes $u$ such that $M_u > M_v$. Then, there are $n - M_v + \pi_M(v)$ nodes $w$ such that $m_w > M_v$.  Similarly,  there are $n- \pi_m(v)$ nodes $w$ satisfying $m_w \geq m_v$, and then $n - m_v + \pi_m(v)$ nodes $u$ such that $M_u > m_v$.
 
 Then, in $[m_v, M_v]$ there are  $\pi_m(v) + \pi_M(v) - m_v$ nodes $u$ such that $M_u \in [m_v, M_v]$ and $M_v - \pi_m(v) - \pi_M(v)$ nodes $w$ such that $m_w \in [m_v, M_v]$. For each $v\in V$, let us call $S_s$ be the set of nodes $z$ such that $[m_z, M_z]\subset [m_v, M_v]$.

Suppose that $\{[m_v, M_v]\}_{v \in V}$ is a proper model for $G$, and for each $v\in V$. Then, we have that $n_M(v) + |S_v| = \pi_m(v) + \pi_M(v) - m_v$ and $n_m(v) + |S_v| = M_v - \pi_m(v) - \pi_M(v)$. We deduce that $\pi_m(v) + \pi_M(v) - m_v - n_M(v) = M_v - \pi_m(v) - \pi_M(v) - n_m(v)$, from which we deduce $2(\pi_M(v) + \pi_m(v)) = M_v + m_v + n_M(v) - n_m(v)$.

Suppose now that $\{[m_v, M_v]\}_{v \in V}$ is not a proper model for $G$. Then, there exist a pair of nodes $\{u, v\}  \notin E$ such that $[m_u,M_u] \cap [m_v \cap M_v] \neq \emptyset$ with the intervals not containing each other. From all such pairs, let us pick one such that $m_v$ is minimum. Then necessarily there exists $m_\omega \in [m_v, M_v]$, $M_\omega > M_v$ and $\{v,\omega\}\notin E$. This implies that $n_m(v) + |S_v| < M_v - \pi_m(v) - \pi_M(v)$. On the other hand, the choice of $v$ implies that all nodes such $u$ such that $M_u \in [m_v, M_v]$ either belong to $N(v)$ or belong to $S_v$. Then $n_M(v) + |S_v| =   \pi_m(v) + \pi_M(v) - m_v$. We deduce that $2(\pi_M(v) + \pi_m(v)) \neq M_v + m_v + n_M(v) - n_m(v)$.
\end{proof}

Now we are ready to give the main result of this section.

\begin{theorem}
	$\circlegraph$ belongs to $\dMAM[\mathcal{O}(\log n)]$.
\end{theorem}

\begin{proof}
Consider the protocols for \sizeofG, \permutation\ and \COP of Propositions \ref{prop:sizeofG}, \ref{prop:permu} and \ref{prop:COP}. Given an instance $\langle G, \id\rangle$, consider the following  three-round protocol. In the first round, the prover provides each node $v$ with the following information:
\begin{enumerate}
\item The certification of the total number of nodes $n$, according to the protocol for $\sizeofG$.
\item The limits of the interval $I_v = [m_v, M_v]$, and the values $\pi_m(v)$ and $\pi_M(v)$.
\item The certification  of $\{m_v\}_{v\in V} \cup \{M_v\}_{v\in V} = [2n]$ according to the protocol for $\permutation$.
\item The certification of the correctness of $\{(m_v, \pi_m(v))\}_{v\in V}$ according to the protocol for \COP.
\item The certification of the correctness of $\{(M_v, \pi_M(v))\}_{v\in V}$ according to the protocol for \COP.
\end{enumerate}

Then, in the second and third round the nodes perform the remaining interactions of the protocols for $\permutation$ and \textsc{corresponding} \textsc{order}. In the verification round, the nodes first check the correctness of \lipics{1}-\lipics{5} according to the verification rounds for $\sizeofG$, $\permutation$ and \textsc{corresponding} \textsc{order}. Then, each node $v$ computes $n_m(v)$ and $n_M(v)$, and checks the following conditions:

\begin{itemize}
\item[\lipics{a.}] For every $u \in N(v)$, $I_v \cap I_u$, $I_v \not\subseteq I_u$ and $I_u \not\subseteq I_v$.
\item[\lipics{b.}] $2(\pi_M(v) + \pi_m(v)) = M_v + m_v + n_M(v) - n_m(v)$.
\end{itemize}

We now analyze completeness and soundness. \\

\noindent \textbf{Completeness:} Suppose that input graph $G$ is a circle graph. Then Merlin gives a proper model $\{[m_v,M_v]\}_{v\in V}$  for $G$. Merlin also provides the correct orders $\{\pi_m(v)\}_{v\in V}$ and $\{\pi_M(v)\}_{v\in V}$, the correct number of nodes $n$, and the certificates required in the sub-routines. Then, the nodes  verify correctness of \lipics{1}-\lipics{5} with probability greater than $2/3$, by the correctness of the protocols for $\sizeofG$, $\permutation$ and $\COP$. Finally, condition \lipics{a} is verified by definition of a proper model, and condition \lipics{b} is verified by \Cref{lem:circchar}. We deduce that every node accepts with probability greater than $2/3$. \\
	
\noindent \textbf{Soundness:} Suppose now that $G$ is not a circle graph. By the soundness of the protocols for $\sizeofG$, $\permutation$ and $\COP$, we know that at least one node rejects the certificates not satisfying \lipics{1}-\lipics{5}, with probability greater than $2/3$. Suppose then that conditions \lipics{1}-\lipics{5} are verified. Observe that each set of intervals satisfying condition \lipics{a} form a semi-proper model for $G$. Since $G$ is not a circle graph, by  \Cref{lem:circchar} we deduce that at least one node fails to verify \lipics{a} or \lipics{b}. All together, we deduce that at least one node rejects with probability greater than $2/3$. \\

  We now analyze the communication complexity of the protocol:  the certification for $\sizeofG$, $\permutation$ and $\COP$ is $\cO(\log n)$, given by  \Cref{prop:sizeofG}, \ref{prop:permu} and \ref{prop:COP}, and for each $v\in V$, the values $m_v, M_v, \pi_m(v), \pi_M(v)$ can be encoded in $\cO(\log n)$ as they are numbers in $[2n]$ or $[n]$. Overall the total communication is $\cO(\log n)$. 
\end{proof}


\section{Polygon Circle Graphs} \label{sec:polyg}

In this section, we give a three-round protocol for the recognition of polygon-circle graph. This extension is based in a non-trivial extension of the properties of circle graphs. 

Remember that a $n$-node graph $G = (V,E)$ is a $k$-polygon-circle graph if and only if $G$ is the intersection model of a set of $n$ polygons of $k$ vertices inscribed in a circle, namely $\{P_v\}_{v\in V}$.  Further, every $k$-polygon-circle graph admits a model satisfying the following conditions \cite{Brandstdt1999}: (1) for each $v \in V$, the polygon $P_v$  is represented as a set of $k$ vertices $\{p_0(v), \dots, p_{k-1}(v)\}$ such that, for each $i\in[k-1]$, $1\leq p_i(v) < p_{i+1}(v) \leq n\cdot k $, and (2)
$\bigcup_{v\in V}\bigcup_{i \in [k]} \{p_i(v)\} = [n\cdot k]$. In other words, each value in $[k\cdot n]$ corresponds to a unique vertice of some polygon. A set of polygons satisfying conditions (1) and (2) are called a \emph{proper polygon model} for $G$. Similar to previous cases, when we just ask that adjacent nodes have intersecting polygons (but not necessarily the reciprocal) we say that the model is a \emph{semi-proper polygon model} for $G$. 

Let $G$ be a graph and $\{P_v\}$ be a semi-proper model for $G$. For each $v \in V$. Let us call $\alpha(v)$ the set of points in $\{1, \dots, p_1(v)\} \cup \{p_k(v), \dots, kn\}$ that do not correspond to a neighbor of $v$. For each $i\in [k]$, we also call $\beta_i(v)$ the set of nodes $w \notin N(v)$ such that $p_i(w) \in \alpha(v)$. Formally,

$$ \alpha(v) = \{ i \in [0, p_1(v)]_\N \cup [p_k(v), kn-1]_\N: \forall u \in N(v), i \notin P_u\}$$ 
$$ \beta_i(v) = \{w \in V: p_i(w) \in \alpha(v)\}$$

\begin{lemma}\label{lem:knrepre}
Let $G=(V,E)$ be a graph, and let $\{P_v\}_{v\in V}$ a semi-proper model for $G$. Then $\{P_v\}_{v\in V}$ is a proper model for $G$ if and only if $|\alpha(v)|  =  k|\beta_1(v)|$ for every $v\in V$.
\end{lemma}
\begin{proof}
 Let us suppose first that $\{P_v\}_{v\in V}$ is a proper model for $G$ and $v$ be an arbitrary node. If $\alpha(v)=\emptyset$ the result is direct. Then, let us suppose that $\alpha(v)\neq \emptyset$, and let us pick $q \in \alpha(v)$. Then necessarily there exists $i \in [k]$ such that $q$ belongs to $\beta_i(v)$. Observe that for each node $w$ in $\beta_i(v)$, all the vertices of the polygon $P_w$ are contained $\alpha_v$. Otherwise, the polygons $P_w$ and $P_v$ would have non-empty intersection, which contradicts the fact that $\{P_v\}_{v\in V}$ is a proper model. This implies that $|\alpha(v)| = k|\beta_i(v)|$, for every $i\in [k]$. In particular $|\alpha(v)| = k|\beta_1(v)|$.

Let us suppose now that $\{P_v\}_{v\in V}$ is not a proper model for $G$. Let us define the set $C$ of vertices having non-neighbor with intersecting polygons, formally $$C = \{v\in V : \exists w\in V, \{v,w\} \not\in E\text{ and }P_v\cap P_w \neq \emptyset\}.$$  Now pick $v\in C$ such that $p_1(v)$ is maximum, and call $C_v$ the set of non-neighbors of $v$ whose polygons intersect with $P_v$. Let $w$ be an arbitrary node in $C_v$.  By the maximality of $p_1(v)$, we know that $p_1(w) \in \beta_1(v)$. But since $P_w \cap P_v \neq \emptyset$, there must exist $i \in [k]$ such that $p_i(w) \notin \beta_i(v)$.  This implies that  $|\beta_1(v)| \geq |\beta_i(v)|$ for every $i\in [k]$, and at least one of these inequalities is strict. Since $|\alpha(v)| = \sum_{i \in [k]} |\beta_i(v)|$, we deduce that $k|\beta_1(v)| > |\alpha(v)|$.
\end{proof}

Let $G = (V,E)$ be a graph, and $\{P_v\}_{v\in V}$ be  a semi-proper polygon model for $G$. For each $i \in [k]$ and $v \in V$, we denote by $\pi_i(v)$ the cardinality of the set $\{u \in V: p_i(u) < p_i(v)\}$, and denote by $\sigma_i(v)$ the cardinality of the set  $\{q < p_i(v): \exists u \in V: p_1(u) = q \vee p_k(u) = q\}$.For a node $v$, we denote $N_{1,k}(v)$ the number vertices of polygons corresponding to neighbors of $v$, that are contained $[0, p(v_1)]_\N \cup [p(v_k), kn]_\N$. Formally, $$N_{1,k}(v) = |\{q \in [0, p(v_1)]_\N \cup [p(v_k), kn]_\N : \exists w \in N(v), q \in P_w\}|$$

\begin{lemma}\label{lem:polytechnical} Let $G = (V,E)$ be a graph, and $\{P_v\}_{v\in V}$ be  a semi-proper polygon model for $G$. Then, 
	\begin{align*}
		|\alpha(v)| &= kn  - p_k(v) + p_1(v) - 1  - N_{1,k}(v)\\
		|\beta_1(v)| &= n - \sigma_k(v) + \pi_k(v) + \pi_1(v)
	\end{align*}
	
\end{lemma}

\begin{proof}
Let  $\{P_v\}_{v\in V}$ be  a semi-proper polygon model for $G$ and $v$ be an arbitrary node. First, observe that there are $p_1(v)$ integer positions for vertices in $[p_1(v)]$ and $(kn-1) - p_k$ positions for vertices in $[p_k, kn-1]$. Then, there are $kn - p_k + p_1(v) - 1$ available integer positions in $[p_1(v)] \cup [p_k(v),kn-1]_\N$. Since $N_{1,k}(v)$ of these positions are occupied by a polygon corresponding some neighbor of $v$, we deduce that $\alpha(v) = kn - p_k + p_1(v) - 1 - N_{1,k}(v)$.

Second, observe that the set $\{p_1(u), p_k(u)\}_{u \in V}$ uses $2n$ of the $kn$ possible positions. Then, there are $2n - \sigma_k(v)$ positions used the elements of $\bigcup_{u\in V}\{p_1(u), p_k(u)\} \cap [p_k(v)+1, kn-1]_\N$. On the other hand, there are $n- \pi_k$ positions used by vertices in $\bigcup_{u\in V}\{p_k(u)\} \cap [p_k(v)+1, kn-1]_\N$. Therefore, there are $n - \sigma_k(v) + \pi_k(v)$ positions used by $\bigcup_{u\in V}\{p_1(u)\} \cap [p_k(v)+1, kn-1]_\N$. Finally, noticing that there are $\pi_1(v)$ positions used by $\bigcup_{u\in V}\{p_1(u)\} \cap [0, p_1(v)-1]_\N$, we deduce that $|\beta_1(v)| = n - \sigma_k(v) + \pi_k(v) + \pi_1(v)$.
\end{proof}

We are now ready to give the main result of this section.

\begin{theorem}
	$\polygoncircle$ belongs to $\dMAM[\log n]$.
\end{theorem}

\begin{proof}
	Consider the protocols for \sizeofG, \permutation\ and \COP of Propositions \ref{prop:sizeofG}, \ref{prop:permu} and \ref{prop:COP}.Given an instance $\langle G, \id\rangle$, consider the following  three round protocol. In the first round, the prover provides each node $v$ with the following information:
\begin{enumerate}
\item The certification of the total number of nodes $n$, according to the protocol for $\sizeofG$.
\item The vertices of the polygon $P_v$, denoted $V(P_v) = \{p_1(v), \dots, p_k(v)\}$.
\item The values of $\pi_1(v)$, $\pi_k(v)$ and $\sigma(v)$. 
\item The certification  of $\bigcup_v V(P_v) = [k\cdot n]$ according to the protocol for $\permutation$.
\item The certification of the correctness of $\{(p_1(v), \pi_1(v))\}_{v\in V}$ according to the protocol for \COP.
\item The certification of the correctness of $\{(p_k(v), \pi_k(v))\}_{v\in V}$ according to the protocol for \COP.
\item The certification of the correctness of $\{(p_1(v), \sigma_1(v)\}_{v\in V}$ and the collection $\{p_k(v), \sigma_k(v)\}_{v\in V}$ according to the protocol for \COP.
\end{enumerate}
	
	Then, in the second and third round the nodes perform the remaining interactions of the protocols for $\permutation$ and \COP. In the verification round, the nodes first check the correctness of \lipics{1}-\lipics{7} according to the verification rounds for $\sizeofG$, $\permutation$ and \COP. \\
	
	{\bf Remark}: in order to check \lipics{7}, each node has to play the role of two different nodes $v', v''$, one to verify $\{(p_1(v), \sigma(v')\}_{v\in V}$, and the other one to verifies $\{(p_k(v), \sigma(v'')\}_{v\in V}$, where $\sigma(v') = \sigma_1(v)$ and $\sigma(v'') = \sigma_k(v)$. To do so, Merlin gives  $v$ the certificates of $v'$ and $v''$, and $v$ answers with the random bits as if they would be generated by $v'$ and $v''$. Obviously, this increases the communication cost by a factor of $2$.\\
	
	Then, each node $v$ computes $|\beta(v)|$ and $|\alpha(v)|$ according to the expressions of \cref{lem:polytechnical}, and checks the following conditions:
	\begin{itemize}
	\item[\lipics{a.}] $\forall u \in N(v)$, $P_u \cap P_v \neq \emptyset$.
	\item[\lipics{b.}] $|\alpha(v)| = k|\beta_1(v)|$.
	\end{itemize}
	
We now analyze completeness and soundness. \\

\noindent \textbf{Completeness:} Suppose that input graph $G$ is a $k$-polygon-circle graph. Then Merlin gives a proper polygon model $\{P_v\}_{v\in V}$  for $G$. Merlin also provides  the correct number of nodes $n$, correct orders $\{\pi_1(v)\}_{v\in V}$ and $\{\pi_k(v)\}_{v\in V}$,  $\{\sigma_1(v)\}_{v\in V}$ and  $\{\sigma_k(v)\}_{v\in V}$, and the certificates required in the corresponding sub-routines. Then, the nodes verify correctness of \lipics{1}-\lipics{7} with probability greater than $2/3$, by the correctness of the protocols for $\sizeofG$, $\permutation$ and $\COP$. Finally, condition \lipics{a} is verified by definition of a proper model, and condition \lipics{b} is verified by \cref{lem:knrepre}. We deduce that every node accepts with probability greater than $2/3$. \\
		
\noindent \textbf{Soundness:} Suppose now that $G$ is not a $k$-polygon-circle graph. By the soundness of the protocols for $\sizeofG$, $\permutation$ and \textsc{corresponding} \textsc{order}, we now that at least one node rejects the certificates not satisfying \lipics{1}-\lipics{7}, with probability greater than $2/3$. Suppose then that conditions \lipics{1}-\lipics{7} are verified. Observe that every set of polygons satisfying condition \lipics{a} form a semi-proper polygon model for $G$. Since $G$ is not a $k$-polygon-circle graph, by  \cref{lem:knrepre} we deduce that at least one node fails to verify \lipics{a} or \lipics{b}. All together, we deduce that at least one node rejects with probability greater than $2/3$. \\

  We now analyze the communication complexity of the protocol:  the certification for $\sizeofG$, $\permutation$ and $\COP$ is $\cO(\log n)$, given by  \cref{prop:sizeofG}, \ref{prop:permu} and \ref{prop:COP}. On the other hand, for each $v\in V$, the values $\pi_1(v), \pi_k(v), \sigma_1(v), \sigma_2(v), V(P_v)$ can be encoded in $\cO(\log n)$ as they are numbers in $[n]$, $[2n]$ or $[kn]$. Overall the total communication is $\cO(\log n)$. 
\end{proof}


\section{Lower Bounds} \label{sec:lower}

In this section we give logarithmic lower-bounds in the certificate sizes of any proof labeling scheme that recognizes the class of permutation, trapezoid, circle or polygon-circle graphs. In order to so, we use a technique given by Fraigniaud et al \cite{fraigniaud2019randomized}, called \emph{crossing edge}, and which we detail as follows. Let $G=(V,E)$ be a graph and let $H_1 = (V_1,E_1)$ and $H_2 = (V_2,E_2)$ be two subgraphs of $G$. We say that $H_1$ and $H_2$ are independent if and only if $V_1\cap V_2 = \emptyset$ and $E\cap (V_1\times V_2) = \emptyset$.

\begin{definition}[\cite{fraigniaud2019randomized}]
	Let $G=(V,E)$ be a graph and let $H_1 = (V_1,E_1)$ and $H_2 = (V_2,E_2)$ be two independent isomorphic subgraphs of $G$ with isomorphism $\sigma\colon V_1\to V_2$. The {\sc crossing} of $G$ induced by $\sigma$, denoted by $\sigma_{\bowtie}(G)$, is the graph obtained from $G$ by replacing every pair of edges $\{u,v\}\in E_1$ and $\{\sigma(u),\sigma(v)\}\in E_2$, by the pair $\{u,\sigma(v)\}$ and $\{\sigma(u),v\}$.
\end{definition}

Then, the tool that we use to build our lower-bounds is the following. 

\begin{theorem}[\cite{fraigniaud2019randomized}]\label{teo:lowerbound}
	Let $\mathcal{F}$ be a family of network configurations, and let $\mathcal{P}$ be a boolean predicate over $\mathcal{F}$. Suppose that there is a configuration $G_s\in\mathcal{F}$ satisfying that (1) $G$ contains as subgraphs $r$ pairwise independent isomorphic copies $H_1,...,H_r$ with $s$ edges each, and (2) there exists  $r$ port-preserving isomorphisms $\sigma_i\colon V(H_1)\to V(H_i)$ such that for every $i\neq j$, the isomorphism $\sigma^{ij} = \sigma_i\circ\sigma_j^{-1} $ satisfies $\mathcal{P}(G_s)\neq \mathcal{P}(\sigma^{ij}_{\bowtie}(G)_s)$. Then,  the verification complexity of any proof-labeling scheme for $\mathcal{P}$ and $\mathcal{F}$ is $\Omega\left(\dfrac{log(r)}{s}\right)$.
\end{theorem}

Let us consider first permutation and trapezoid graphs. Let $\mathcal{F}$ the family of instances of $\permutationgraph$, induced by the family of graphs $\{Q_{n}\}_{n>0}$. Each graph $Q_n$ consists of $5n$ nodes forming a path $\{v_1, \dots, v_{5n}\}$ where we add the edge $\{v_{5i-3}, v_{5i-1}\}$, for each $i \in [n]$. It is easy to see that for each $n>0$, $Q_n$ is a permutation graph (and then also a trapezoid graph). In \cref{fig:lowerboundperm1} is depicted the graph $Q_3$ and its corresponding  model.

\begin{figure}[h!]
	\centering
	\input{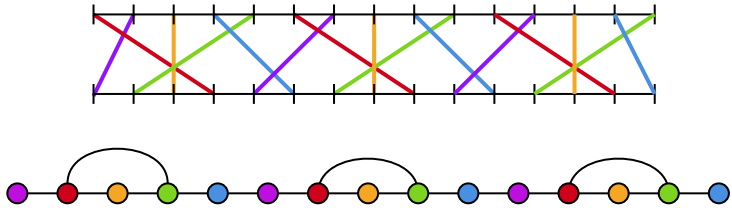}
	\caption{Graph $Q_3$ and a permutation model for $Q_3$. }
	\label{fig:lowerboundperm1}
\end{figure}

Given $Q_n$ defined above, consider the subgraphs $H_i = \{v_{5i-2}, v_{5i-1}\}$, for each $i \in [n]$, and the isomorphism $\sigma_i\colon V(H_1)\to V(H_i)$ such that $\sigma_i(v_3) = v_{5i-2}$ and $\sigma_i(v_4) = v_{5i-1}$. 

\begin{lemma}
For each $i\neq j$, the graph $\sigma^{ij}_{\bowtie}(Q_n)$ it is not a trapezoid graph. 
\end{lemma}

\begin{proof}
Given $i< j$,  observe that in  $\sigma^{ij}\colon V(H_j)\to V(H_i)$, the nodes $v_{5j-3}$, $v_{5j-2}$, $v_{5i-1}$, $v_{5i-3}$, $v_{5i-2}$, $v_{5j-1}$ form an induced cycle of length $6$ (see \cref{fig:lowerboundperm2} for an example). 

\begin{figure}[h!]
	\centering
	\tikzset{every picture/.style={line width=0.75pt}} 

\begin{tikzpicture}[x=0.75pt,y=0.75pt,yscale=-1,xscale=1]

\draw    (190,77.46) -- (240,77.46) ;
\draw    (265,77.46) -- (365,77.46) ;
\draw    (390,77.46) -- (540,77.46) ;
\draw [color={rgb, 255:red, 208; green, 2; blue, 27 }  ,draw opacity=1 ]   (240,78.54) .. controls (240,113.82) and (390,113.82) .. (390,78.54) ;
\draw [color={rgb, 255:red, 208; green, 2; blue, 27 }  ,draw opacity=1 ]   (265,78.54) .. controls (265.11,98.71) and (365.11,99.04) .. (365,78.54) ;
\draw    (215,72.46) .. controls (215.11,48.85) and (265.11,49.51) .. (265,72.46) ;
\draw    (340,77.46) .. controls (340.11,53.85) and (390.11,54.51) .. (390,77.46) ;
\draw  [fill={rgb, 255:red, 245; green, 166; blue, 35 }  ,fill opacity=1 ] (235,77.46) .. controls (235,74.7) and (237.24,72.46) .. (240,72.46) .. controls (242.76,72.46) and (245,74.7) .. (245,77.46) .. controls (245,80.22) and (242.76,82.46) .. (240,82.46) .. controls (237.24,82.46) and (235,80.22) .. (235,77.46) -- cycle ;
\draw  [fill={rgb, 255:red, 126; green, 211; blue, 33 }  ,fill opacity=1 ] (260,77.46) .. controls (260,74.7) and (262.24,72.46) .. (265,72.46) .. controls (267.76,72.46) and (270,74.7) .. (270,77.46) .. controls (270,80.22) and (267.76,82.46) .. (265,82.46) .. controls (262.24,82.46) and (260,80.22) .. (260,77.46) -- cycle ;
\draw  [fill={rgb, 255:red, 208; green, 2; blue, 27 }  ,fill opacity=1 ] (210,77.46) .. controls (210,74.7) and (212.24,72.46) .. (215,72.46) .. controls (217.76,72.46) and (220,74.7) .. (220,77.46) .. controls (220,80.22) and (217.76,82.46) .. (215,82.46) .. controls (212.24,82.46) and (210,80.22) .. (210,77.46) -- cycle ;
\draw  [fill={rgb, 255:red, 74; green, 144; blue, 226 }  ,fill opacity=1 ] (285,77.46) .. controls (285,74.7) and (287.24,72.46) .. (290,72.46) .. controls (292.76,72.46) and (295,74.7) .. (295,77.46) .. controls (295,80.22) and (292.76,82.46) .. (290,82.46) .. controls (287.24,82.46) and (285,80.22) .. (285,77.46) -- cycle ;
\draw  [fill={rgb, 255:red, 189; green, 16; blue, 224 }  ,fill opacity=1 ] (310,77.46) .. controls (310,74.7) and (312.24,72.46) .. (315,72.46) .. controls (317.76,72.46) and (320,74.7) .. (320,77.46) .. controls (320,80.22) and (317.76,82.46) .. (315,82.46) .. controls (312.24,82.46) and (310,80.22) .. (310,77.46) -- cycle ;
\draw  [fill={rgb, 255:red, 208; green, 2; blue, 27 }  ,fill opacity=1 ] (335,77.46) .. controls (335,74.7) and (337.24,72.46) .. (340,72.46) .. controls (342.76,72.46) and (345,74.7) .. (345,77.46) .. controls (345,80.22) and (342.76,82.46) .. (340,82.46) .. controls (337.24,82.46) and (335,80.22) .. (335,77.46) -- cycle ;
\draw  [fill={rgb, 255:red, 189; green, 16; blue, 224 }  ,fill opacity=1 ] (185,77.46) .. controls (185,74.7) and (187.24,72.46) .. (190,72.46) .. controls (192.76,72.46) and (195,74.7) .. (195,77.46) .. controls (195,80.22) and (192.76,82.46) .. (190,82.46) .. controls (187.24,82.46) and (185,80.22) .. (185,77.46) -- cycle ;
\draw  [fill={rgb, 255:red, 245; green, 166; blue, 35 }  ,fill opacity=1 ] (360,77.46) .. controls (360,74.7) and (362.24,72.46) .. (365,72.46) .. controls (367.76,72.46) and (370,74.7) .. (370,77.46) .. controls (370,80.22) and (367.76,82.46) .. (365,82.46) .. controls (362.24,82.46) and (360,80.22) .. (360,77.46) -- cycle ;
\draw  [fill={rgb, 255:red, 126; green, 211; blue, 33 }  ,fill opacity=1 ] (385,77.46) .. controls (385,74.7) and (387.24,72.46) .. (390,72.46) .. controls (392.76,72.46) and (395,74.7) .. (395,77.46) .. controls (395,80.22) and (392.76,82.46) .. (390,82.46) .. controls (387.24,82.46) and (385,80.22) .. (385,77.46) -- cycle ;
\draw  [fill={rgb, 255:red, 74; green, 144; blue, 226 }  ,fill opacity=1 ] (410,77.46) .. controls (410,74.7) and (412.24,72.46) .. (415,72.46) .. controls (417.76,72.46) and (420,74.7) .. (420,77.46) .. controls (420,80.22) and (417.76,82.46) .. (415,82.46) .. controls (412.24,82.46) and (410,80.22) .. (410,77.46) -- cycle ;
\draw    (465,77.46) .. controls (465.11,53.85) and (515.11,54.51) .. (515,77.46) ;
\draw  [fill={rgb, 255:red, 189; green, 16; blue, 224 }  ,fill opacity=1 ] (435,77.46) .. controls (435,74.7) and (437.24,72.46) .. (440,72.46) .. controls (442.76,72.46) and (445,74.7) .. (445,77.46) .. controls (445,80.22) and (442.76,82.46) .. (440,82.46) .. controls (437.24,82.46) and (435,80.22) .. (435,77.46) -- cycle ;
\draw  [fill={rgb, 255:red, 208; green, 2; blue, 27 }  ,fill opacity=1 ] (460,77.46) .. controls (460,74.7) and (462.24,72.46) .. (465,72.46) .. controls (467.76,72.46) and (470,74.7) .. (470,77.46) .. controls (470,80.22) and (467.76,82.46) .. (465,82.46) .. controls (462.24,82.46) and (460,80.22) .. (460,77.46) -- cycle ;
\draw  [fill={rgb, 255:red, 245; green, 166; blue, 35 }  ,fill opacity=1 ] (485,77.46) .. controls (485,74.7) and (487.24,72.46) .. (490,72.46) .. controls (492.76,72.46) and (495,74.7) .. (495,77.46) .. controls (495,80.22) and (492.76,82.46) .. (490,82.46) .. controls (487.24,82.46) and (485,80.22) .. (485,77.46) -- cycle ;
\draw  [fill={rgb, 255:red, 126; green, 211; blue, 33 }  ,fill opacity=1 ] (510,77.46) .. controls (510,74.7) and (512.24,72.46) .. (515,72.46) .. controls (517.76,72.46) and (520,74.7) .. (520,77.46) .. controls (520,80.22) and (517.76,82.46) .. (515,82.46) .. controls (512.24,82.46) and (510,80.22) .. (510,77.46) -- cycle ;
\draw  [fill={rgb, 255:red, 74; green, 144; blue, 226 }  ,fill opacity=1 ] (535,77.46) .. controls (535,74.7) and (537.24,72.46) .. (540,72.46) .. controls (542.76,72.46) and (545,74.7) .. (545,77.46) .. controls (545,80.22) and (542.76,82.46) .. (540,82.46) .. controls (537.24,82.46) and (535,80.22) .. (535,77.46) -- cycle ;

\end{tikzpicture}
	\caption{Graph $\sigma^{12}_{\bowtie}(Q_3)$, where in red are represented the crossing edges. Observe that this graph is not a trapezoid graph, as it contains an induced cycle of length $6$.}
	\label{fig:lowerboundperm2}
\end{figure}
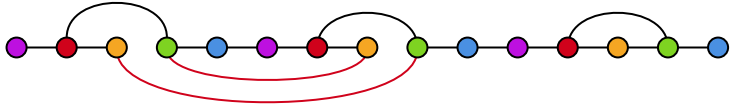
	
	As a trapezoid graph have induced cycles of length at most $6$, we deduce that $\sigma^{ij}_{\bowtie}(Q_n)$ is not a trapezoid graph.
\end{proof}

By \cref{teo:lowerbound} and the abode result, the lower bound result is direct.

\begin{theorem}
	Any proof labeling scheme for \permutationgraph~or \trapezoid~has proof-size of $\Omega\left(\log n \right)$ bits.
\end{theorem}

We now tackle the lower-bound for circle and polygon-circle graphs. Let $\mathcal{G}$ the family of instances of $\circlegraph$, defined by the family of graphs $\{M_{n}\}_{n>2}$. Each graph $M_n$ consists  of $6n$ nodes, where $4n$ nodes form a path $\{v_1, \dots, v_{4n}\}$ where we add, for each $i \in [n]$, the edges $\{v_{4i-3}, v_{4n+i}\}$, $\{v_{4i-2}, v_{5n+i}\}$ and $\{v_{4n+i}, v_{5n+i}\}$. It is easy to see that for each $n>0$, $M_n$ is a circle graph (and then also a polygon-circle graph). In Figure \ref{fig:lowerboundcircle1} is depicted the graph $M_4$ and its corresponding  model.

\begin{figure}[h!]
	\centering
	\input{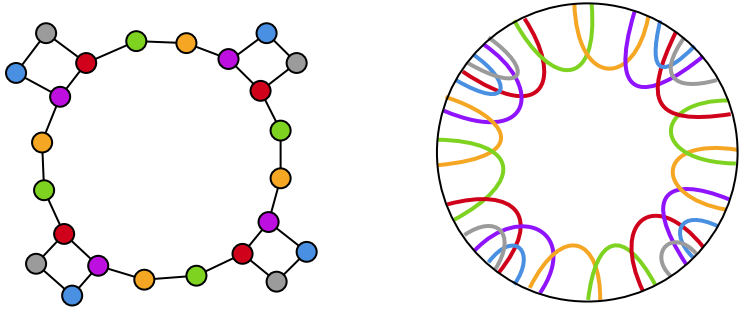}
	\caption{Graph $M_4$ and a permutation model for $M_4$. }
	\label{fig:lowerboundcircle1}
\end{figure}

Given $M_n$ defined above, consider the subgraphs $H_i = \{v_{4n+i}, v_{5n+i}\}$, for each $i \in [n]$, and the isomorphism $\sigma_i\colon V(H_1)\to V(H_i)$ such that $\sigma_i(4n+1) = v_{5n+i}$ and $\sigma_i(5n+1) = v_{4n+i}$. 

\begin{lemma}\label{lem:lowerpoly}
For every $k>0$ and each $i\neq j$, the graph $\sigma^{ij}_{\bowtie}(M_n)$ it is not a $k$-polygon-circle graph. 
\end{lemma}

\begin{proof}

First, observe that in  $\sigma_{\bowtie}^{ij}(M_n)$ we have two induced cycles defined by $C_1 = v_1, \dots, v_{4n}$, and $C_2 = v_{4j-3}, v_{4n+j}, v_{5n+i}, v_{4i-2}, v_{4i-3}, v_{4n+i}, v_{5n+j}$.  Moreover $|V(C_1)\cap V(C_2)| =  4$, $|V(C_1) - V(C_2)|  = 4n-4$ and $|V(C_2) - V(C_1)| = 4$. See \cref{fig:lowerboundcircle2} for a representation of $\sigma^{i,i+1}_{\bowtie}(M_4)$.
	
\begin{figure}[h!]
	\centering
	\input{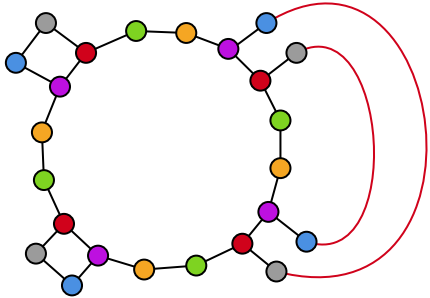}
	\caption{Graph $\sigma^{i,i+1}_{\bowtie}(M_4)$, where in red are represented the crossing edges.}
	\label{fig:lowerboundcircle2}
\end{figure}

	\begin{claim}\label{claim:1}
		Every graph $G$ consisting in two graphs $C_1$ and $C_2$ such that $|V(C_1)\cap V(C_2)|\geq 4$, $|V(C_1) - V(C_2)|\geq 2$ and $|V(C_2) - V(C_1)|\geq 2$ is not a $k$-polygon-circle graph, for every $k>0$.
	\end{claim}

	 \begin{proof}
		Let us denote by $v_i$ and $v_f$ two special nodes with degree $3$, connecting the two cycles.
		Suppose there exists a $k$-polygon-circle model for $G$. Observe that, if we delete all polygons corresponding to nodes of $V(C_2) - V(C_1)$, we obtain a polygon model for $C_1$. However, the cycle $C_1$ has at least $4$ nodes, because $|V(C_1) - V(C_2)|\geq 2$ and $|V(C_1)\cap V(C_2)|\geq 4$. Then, there is no way to add the removed polygons corresponding to $V(C_2)- V(C_1)$, without intersecting a polygon of $V(C_1)-V(C_2)$.
		
	\end{proof}

	Then, by \cref{claim:1}, we deduce that the graph induced by $C_1 \cup C_2$ is not a $k$-polygon-cycle graph.  Since the class of polygon-circle graphs is hereditary, we deduce that $\sigma^{ij}_{\bowtie}(M_n)$ it is not a $k$-polygon-circle graph. 
\end{proof}

Direct by Theorem \ref{teo:lowerbound} and Lemma \ref{lem:lowerpoly} we deduce the following result.  

\begin{theorem}
	Any proof labeling scheme for \circlegraph~or \polygoncircle~has proof-size $\Omega(log(n))$ bits.
\end{theorem}

\bibliography{biblio}

%
%
%
%
%
%
%
%
%
%
\end{document}